\documentclass[a4paper,11pt,DIV=9]{scrartcl}

\usepackage[T1]{fontenc}
\usepackage{amsmath}
\usepackage{amssymb}
\addtokomafont{disposition}{\rmfamily\normalfont\scshape}
\usepackage{mathtools}
\usepackage[osf,sc]{mathpazo}
\usepackage{relsize}
\usepackage{amsthm}
\usepackage{enumerate}
\usepackage[autostyle=true]{csquotes}
\usepackage{verbatim}
\usepackage{tikz}
\usepackage{tikz-cd}
\usepackage{microtype}
\usepackage{pdfpages}
\usepackage[british]{babel}
\usepackage[hidelinks,unicode]{hyperref}
\recalctypearea

\numberwithin{equation}{section}
\newtheorem{theorem}[equation]{Theorem}
\newtheorem*{theorem*}{Theorem}
\newtheorem{corollary}[equation]{Corollary}
\newtheorem{lemma}[equation]{Lemma}
\newtheorem{proposition}[equation]{Proposition}

\theoremstyle{definition}
\newtheorem{definition}[equation]{Definition}
\newtheorem*{definition*}{Definition}

\theoremstyle{remark}
\newtheorem{example}[equation]{Example}

\usepackage{thmtools, thm-restate}
\usepackage{booktabs}
\usepackage{subcaption}
\usepackage{bm}
\usepackage{soul}

\usepackage[boldmath]{numprint}
\npthousandsep{\,}

\usepackage{algorithm}
\usepackage{algpseudocode}

\newcommand{\pr}{\mathbb P}
\newcommand{\expected}{\mathbb E}

\newcommand{\Mc}{\mathcal{M}}
\newcommand{\Pm}{\mathbf{P}}
\newcommand{\St}{S}
\newcommand{\init}{s_{in}}
\newcommand{\Lab}{\textit{Obs}}

\newcommand{\strategy}{\sigma}
\newcommand{\MDP}{\mathsf{M}}
\newcommand{\restartMDP}{\mathsf{M_{r}}}
\newcommand{\restartMDPst}[1]{\mathsf{M}_{\mathsf{r}}^{#1}}
\newcommand{\restact}{\mathsf{r}}
\newcommand{\tail}[1]{\textit{tail}({#1})}
\newcommand{\Actions}{A}
\newcommand{\transitions}{\Delta}
\newcommand{\Distributions}{\mathcal{D}}
\newcommand{\finitepath}{\pi}

\newcommand{\dra}{\mathcal{A}}
\newcommand{\Mcdra}{\Mc \otimes \dra}
\newcommand{\draS}{Q}

\newcommand{\draTr}{\gamma}
\newcommand{\draInit}{q_0}
\newcommand{\draAcc}{Acc}

\newcommand{\run}{\rho}
\newcommand{\mypath}{\pi}

\newcommand{\runs}{\mathsf{Runs}}

\newcommand{\cone}{\mathsf{Cone}}

\newcommand{\trerr}[1]{\xi}
\newcommand{\mperr}[1]{\zeta}

\newcommand{\rem}[1]{}

\usepackage{array}

\newcommand{\acron}[1]{\textsmaller{#1}}
\newcommand{\BSCC}{\acron{BSCC}}
\newcommand{\DRA}{\acron{DRA}}

\newcolumntype{H}{>{\setbox0=\hbox\bgroup}c<{\egroup}@{}}
\newcommand{\Pg}{P_\text{good}}
\newcommand{\Esteps}{\expected[S]}
\newcommand{\Z}{\mathbb{Z}}
\newcommand{\N}{\mathbb{N}}
\newcommand{\bigO}{\mathbb{O}}

\newcommand{\rabin}{\mathcal{R}}

\newcommand{\Rm}{R_\textnormal{m}}
\newcommand{\Pmax}{P_\textnormal{m}}

\newcommand{\sumi}{\sum^\infty}

\newcommand{\RgBSSC}{r_\gamma}
\newcommand{\restart}{\#\restact}

\newcommand{\diamondone}{%
	\sbox0{$\lozenge$}%
	\usebox0\kern-.5\wd0\clap{\raisebox{.1ex}{\scalebox{.7}[1]{1}}}\kern.5\wd0%
}

\newcommand\vincent[1]{\noindent{\textcolor{magenta}{[VPG: #1]}}}

\usetikzlibrary{shapes,positioning,arrows,calc,automata,matrix}
\tikzstyle{state}+=[minimum size = 8mm, inner sep=0,outer sep=1]
\tikzset{->,>=stealth'}

\definecolor{wwhite}{gray}{1}
\usepackage[T1]{fontenc}

\title{Black-box Testing Liveness Properties of Partially Observable Stochastic Systems}
\author{Javier Esparza\footnote{Technical University Munich}~~and Vincent Grande\footnote{RWTH Aachen University, funded by the German Research Council (DFG) within Research Training Group 2236 (UnRAVeL)}}
\date{\vspace{-5ex}}
\begin{document}
	\maketitle
\begin{abstract}
We study black-box testing for stochastic systems and arbitrary $\omega$-regular specifications, explicitly including liveness properties.  We are given a finite-state probabilistic system that we can only execute from the initial state. We have no information on the number of reachable states, or on the probabilities; further, we can only partially observe the states. The only action we can take is to restart the system. We design restart strategies guaranteeing that, if the specification is violated with non-zero probability, then w.p.1 the number of restarts is finite, and the infinite run executed after the last restart violates the specification. This improves on previous work that required full observability. We obtain asymptotically optimal upper bounds on the expected number of steps until the last restart. We conduct experiments on a number of benchmarks, and show that our strategies allow one to find violations in Markov chains much larger than the ones considered in previous work.
\end{abstract}


\newcommand{\lang}{L}
\section{Introduction}
Black-box testing is a fundamental analysis technique when the user does not have access to the design or the internal structure of a system \cite{LeeY96,PeledVY02}. Since it only examines one run of the system at a time, it is computationally cheap, which makes it often the only applicable method for large systems. 

We study the black-box testing problem for finite-state probabilistic systems and $\omega$-regular specifications: Given an $\omega$-regular specification, the problem consists of finding a run of the program that violates the property, assuming that such runs have nonzero probability. 

Let us describe our assumptions in more detail. We do not have access to the code of the system or its internal structure, and we do not know any upper bound on the size of its state space. We can repeatedly execute the system, restarting it at any time. W.l.o.g.~we assume that all runs of the system are infinite.  We do not assume full observability of the states of the system, only that we can observe whether the atomic propositions of the property are currently true or false. For example, if the property states that a system variable, say $x$, should have a positive value infinitely often, then we only assume that at each state we can observe the sign of $x$; letting $\Sigma$ denote the set of possible observations, we have $\Sigma = \{ +, - \}$, standing for a positive and a zero or negative value, respectively (in the rest of the introduction  we shorten ``zero or negative'' to ``negative''). Every system execution induces an observation, that is, an element of $\Sigma^\omega$.  The violations of the property are the $\omega$-words $V \subseteq \Sigma^\omega$  containing only finitely many occurrences of $+$. 

Our goal is to find a \emph{strategy} that decides after each step whether to abort the current run and restart the system, or continue the execution of the current run. The strategy must  ensure that some run that violates the property, that is, a run whose observation belongs to $V$, is eventually executed. The strategy decides depending on the observations made so far. Formally, given $\Sigma$ and the set of actions  $A = \{\restact, \mathsf{c} \}$ (for ``restart'' and ``continue'') a strategy for $V$ is a mapping  from $(\Sigma \times A)^*\Sigma$, the sequence of observations and actions executed so far, to $A$, the next decision.  Our goal is to find a strategy $\sigma$ satisfying the following property:

\begin{quote}
For every finite-state program $P$ over $\Sigma$, if $V\subseteq \Sigma^\omega$ has positive probability and the runs of $P$ are restarted according to $\sigma$, then w.p.1  the number of restarts is finite, and the observation of the run executed after the last restart belongs to $V$.
\end{quote}

Observe that it is not clear that such strategies exist. They are easy to find for safety properties, where the fact that a run violates the property is witnessed by a \emph{finite} prefix\footnote{One can choose for $\sigma$ the strategy ``after the $n$-th reset, execute $n$ steps; if this finite execution is not a witness, restart, otherwise continue forever.'' Indeed, if the shortest witness has length $k$, then for every $n\geq k$, after the $n$-th restart the strategy executes a witness with positive probability, and so it eventually executes one w.p.1.}, but for liveness properties there is no such prefix in general.
We show that these strategies exist for every $\omega$-regular language $V$. Moreover, the strategies only need to maintain a number of counters that depends only on $V$, and not on the program. So in order to restart $P$ according to $\sigma$ one only needs logarithmic memory in the length of the current sequence. 
\begin{example}
To give a first idea of why these strategies also exist for liveness properties, consider the property over $\Sigma = \{ +, -\}$ stating that a variable $x$ should have a positive value only finitely often. The runs violating the property are those that visit $+$-states infinitely often. Our results show that the following strategy works in detecting a run violating the property (among others): 

\begin{quote}
After the $n$-th restart, repeatedly execute blocks of $2n$ steps. If at some point after executing the first block \emph{the second half} of the concatenation of the blocks executed so far contains only negative states, then restart. 
\end{quote}
For example, assume there have been $4$ restarts. Then the strategy repeatedly executes blocks of $8$ steps. If after executing $1,2,3, \ldots$ of these blocks the last $4, 8, 16, \ldots$ states are negative, then the strategy restarts for the $5$th time. If that is never the case, then there are only $4$ restarts. Figure \ref{fig:ideal} shows a family of Markov chains for which naive strategies do not work, but the above strategy does: almost surely the number of restarts is finite and the run after the last restart visits the rightmost state infinitely often. Observe that for every $n \geq 0$ the family exhibits executions that visit $+$ states at least $n$ times, and executions that visit a $+$ state at most once every $n$ steps.

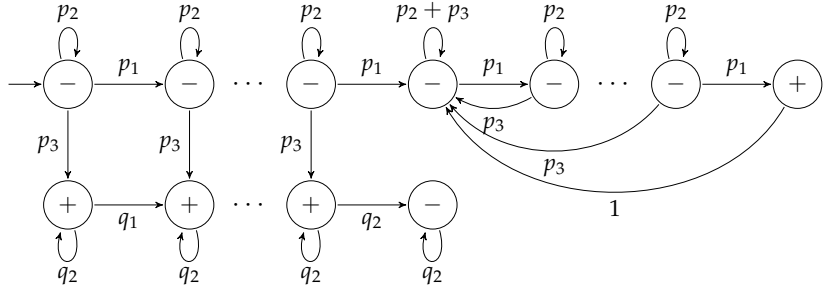
\begin{figure}[t]
	\centering
	\scalebox{0.8}{
		\begin{tikzpicture}
			\def\d{1}
			\node[state,initial,initial text=] (10) at (0,\d){$-$};
			\node[state] (20) at (2,\d){$-$};
			\node (dots) at (3,\d){\large $\cdots$};
			\node[state] (n0) at (4,\d){$-$};
			\node[state] (l1) at (6,\d){$-$};
			\node[state] (l2) at (8,\d){$-$};
			\node (dots) at (9,\d){\large $\cdots$};
			\node[state] (l3) at (10,\d){$-$};
			\node[state] (l4) at (12,\d){$+$};
			\node[state] (21) at (0,-\d){$+$};
			\node[state] (31) at (2,-\d){$+$};
			\node (dots) at (3,-\d){\large $\cdots$};
			\node[state] (41) at (4,-\d){$+$};
			\node[state] (n1) at (6,-\d){$-$};
			
			\path[->]
			(10) edge node[above]{$p_1$} (20)(10) edge node[left]{$p_3$} (21)
			(10) edge[loop above] node[above=-1pt]{$p_2$} ()
			(20) edge[loop above] node[above=-1pt]{$p_2$} ()	
			(20) edge node[left]{$p_3$} (31)	
			(n0) edge node[above]{$p_1$} (l1)
			(n0) edge[loop above] node[above=-1pt]{$p_2$} ()
			(n0) edge node[left]{$p_3$} (41)
			(l1) edge[loop above] node[above=-1pt]{$p_2+p_3$} ()
			(l2) edge[loop above] node[above=-1pt]{$p_2$} ()
			(l3) edge[loop above] node[above=-1pt]{$p_2$} ()
			(l1) edge node[above]{$p_1$} (l2)
			(l3) edge node[above]{$p_1$} (l4)
			(l2) edge[bend left=30] node[below]{$p_3$} (l1)
			(l3) edge[bend left=50] node[below]{$p_3$} (l1)
			(l4) edge[bend left=60] node[below]{$1$} (l1)
			(21) edge node[below]{$q_1$} (31)
			(21) edge[loop below] node[below=-1pt]{$q_2$} ()
			(31) edge[loop below] node[below=-1pt]{$q_2$} ()
			(41) edge[loop below] node[below=-1pt]{$q_2$} ()
			(41) edge node[below]{$q_2$} (n1)
			(n1) edge[loop below] node[below=-1pt]{$q_2$} ();
			
	\end{tikzpicture}}
	\caption{A family of partially observable Markov chains}
	\label{fig:ideal}
\end{figure}
\end{example}
We also obtain asymptotically optimal upper bounds on the expected time until the last restart, that is, on the time until the execution of the run violating the property starts. The bounds depend on two parameters of the Markov chain associated to the program, called  the \emph{progress radius} and the \emph{progress probability}. An important part of our contribution is the identification of these parameters as the key ones to analyze.

While our results are stated in an abstract setting, they easily translate into practice. In a practical scenario, on top of the values of the atomic propositions, we can also observe useful debugging information, like the values of some variables. We let a computer execute runs of the system for some fixed time $t$ according to the strategy $\sigma$. If at time $t$ we observe that the last restart took place a long time ago, then we stop testing and return the run executed since the last restart as candidate for a violation of the property. In the experimental section of our paper we use this scenario to detect errors in \emph{population protocols}, a model of distributed computation, whose state space is too large to find them by other means.

\smallskip\noindent \textbf{Related work.}  There is a wealth of literature on black-box testing and black-box checking \cite{LeeY96,PeledVY02}, but the underlying models are not probabilistic and the methods require to know an upper bound on the number of states. Work on probabilistic model-checking assumes that (a model of) the system is known \cite{Baier2008}. There are also works on black-box verification of probabilistic systems using statistical model checking of statistical hypothesis testing \cite{YounesS02,Sen04,SenVA05,Younes05,YounesCZ10} (see also \cite{LarsenL16a,LegayDB10} for surveys on statistical model checking). They consider a different problem: we focus on producing a counterexample run, while the goal of black-box verification is to accept or reject a hypothesis on the probability of the runs that satisfy a property. 
Our work is also related to the \emph{runtime enforcement problem} \cite{Schneider00,BasinJKZ13,LigattiBW09,FalconeMFR11,FalconeP19}, which also focus on identifying violations of a property. However, in these works either the setting is not probabilistic, or only a subset of the $\omega$-regular properties close to saftey properties is considered.  Finally, the paper closest to ours is \cite{EKKW21}, which considers the same problem, but for fully observable systems. In particular, in the worst case the strategies introduced in \cite{EKKW21} require to store the full sequence of states visited along a run, and so they use linear memory in the length of the current sequence, instead of logarithmic memory, as is the case for our strategy. 

\smallskip\noindent \textbf{Structure of the paper.} The paper is organized as follows. Section \ref{sec:prelims} contains preliminaries. Section \ref{sec:blackboxtesting} introduces the black-box testing problem for arbitrary $\omega$-regular languages with partial observability, and shows that it can be reduced 
to the problem for canonical languages called the Rabin languages. Section \ref{sec:strategy} presents our black-box strategies for the Rabin languages, and proves them correct. Section \ref{sec:quantitative} obtains asymptotically optimal upper bounds on the time to the last restart. Section \ref{sec:experiments} reports some experimental results.


\section{Preliminaries}
\label{sec:prelims}

\textbf{Directed graphs.} A directed graph is a pair $G=(V, E)$, where $V$ is the set of nodes and 
$E \subseteq V\times V$ is the set of edges. A path (infinite path) of $G$ is a finite (infinite) sequence 
$\mypath = v_0, v_1, \ldots$ of nodes such that $(v_i, v_{i+1}) \in E$ for every $i=0,1, \ldots$. A path consisting only of one node is \emph{empty}.
Given two vertices $v, v' \in V$, the \emph{distance} from $v$ to $v'$ is the length of a shortest path from $v$ to $v'$, and the distance from $v$ to a set $V' \subseteq V$ is the minimum over all $v' \in V'$ of the distance from $v$ to $v'$.

A graph $G$ is strongly connected if for every two vertices $v, v'$ there is a path leading from $v$ to $v'$. A graph $G'=(V',E')$ is a subgraph of $G$, denoted $G' \preceq G$, if $V' \subseteq V$ and $E' \subseteq E \cap (V' \times V')$; we write $G' \prec G$ if $G' \preceq G$ and $G'\neq G$. A graph $G' \preceq G$ is a strongly connected component (\acron{SCC}) of $G$ if it is strongly connected and no graph $G''$ satisfying $G' \prec G'' \preceq G$ is strongly connected. An \acron{SCC} $G'=(V',E')$ of $G$  is a bottom \acron{SCC} (\BSCC{}) if $v \in V'$ and $(v, v') \in E$ imply $v' \in V'$.

\medskip\noindent\textbf{Partially observable Markov chains.} Fix a finite set $\Sigma$ of \emph{observations}. A \emph{partially observable Markov chain} is a tuple $\Mc = (\St, \init, \Sigma, \Lab, \Pm)$, where
\begin{itemize}
\item $\Sigma$ is a set of \emph{observations};
\item $\St$ is a finite set of \emph{states} and $\init \in \St$ is the \emph{initial} state;
\item $\Lab \colon \St \to \Sigma$ is an \emph{observation function} that assigns to every state an observation; and
\item $\Pm \;\colon\; \St \times \St \to [0,1]$ is the transition probability matrix, such that for every $s\in \St$ it holds $\sum_{s'\in \St} \Pm(s,s') = 1$,
\end{itemize}
Intuitively, $\Lab(s)$ models the information we can observe when the chain visits $s$. For example, if $s$ is the state of a program,
consisting of the value of the program counter and the values of all variables, $\Lab(s)$ could be just the values of the program counter, or the values of a subset of public variables.
The graph of $\Mc$ has $\St$ as set of nodes and $\{ (s, s') \mid \Pm(s,s') > 0\}$ as set of edges. Abusing language,
we also use $\Mc$ to denote the graph of $\Mc$.
A \emph{run} of $\Mc$ is an infinite path $\run = s_0 s_1 \cdots$ of $\Mc$; we let $\run[i]$ denote the state $s_i$.
The sequence $\Lab(\run) \coloneqq  \Lab(s_0) \Lab(s_1) \cdots$ is the \emph{observation} associated to $\run$. 
Each path $\mypath$ in $\Mc$ determines the set of runs $\cone(\mypath)$ consisting of all runs that start with $\mypath$.
To $\Mc$ we assign the probability space $
(\runs,\mathcal F,\pr)$, where $\runs$ is the set of all runs in $\Mc$, $\mathcal F$ is the $\sigma$-algebra generated by all $\mathsf{Cone}(\mypath)$,
and $\pr$ is the unique probability measure such that
$\pr[\mathsf{Cone}(s_0s_1\cdots s_k)] =
\mu(s_0)\cdot\prod_{i=1}^{k} \Pm(s_{i-1},s_i)$, where the empty product equals $1$.
The expected value of a random variable $f\colon\runs\to\mathbb R$ is $\expected[f]=\int_\runs f\ d\,\pr$.


\medskip\noindent\textbf{Partially Observable Markov Decision Processes.} A \emph{$\Sigma$-observable Markov Decision Process} ($\Sigma$-\acron{MDP}) is a tuple $\MDP = (\St,  \init, \Sigma, \Lab, \Actions, \transitions)$, where $\St,  \init, \Sigma, \Lab$ are as for Markov chains,
$\Actions$ is a finite set of \emph{actions}, and $\transitions \colon  \St \times \Actions \to \Distributions(\St)$ is a \emph{transition function} that for each state $s$ and action $a \in \Actions(s)$ yields a probability 
distribution over successor states. The probability of state $s'$ in this distribution is denoted $\transitions(s, a, s')$.

\medskip\noindent\textbf{Strategies.}  A \emph{strategy} on $\Sigma$-\acron{MDP}s with $A$ as set of actions is a function $\strategy \colon (\Sigma \times \Actions)^*\Sigma \rightarrow \Actions$, which given a finite path $\finitepath = \ell_0 a_0 \, \ell_1 \, a_1 \dots a_{n-1} \, \ell_n \in (\Sigma \times \Actions)^*\Sigma$, yields the action $\strategy(\finitepath) \in A$ to be taken next. Notice that $\strategy$ only ``observes'' $\Lab(s)$, not the state $s$ itself. Therefore,  it can be applied to any $\Sigma$-\acron{MDP} $\MDP=(\St,  \init, \Sigma, \Lab, \Actions, \transitions)$,  inducing the Markov chain $\MDP^\strategy = (\St^\strategy, \init, \Sigma, \Lab, \Actions, \Pm^\strategy)$ defined as follows: $\St^\strategy = (\St \times \Actions)^*\times S$;  and for every state $\finitepath \in \St^\strategy$ of $\MDP^\strategy$ ending at a state $s \in \St$ of $\MDP$, the successor distribution is defined by $\Pm^\strategy(\finitepath, \finitepath \, a \, s') \coloneqq  \transitions(s, a, s')$ if $\strategy(\finitepath)=a$ and $0$ otherwise.


\section{The black-box testing problem}
\label{sec:blackboxtesting}

Fix a set $\Sigma$ of observations,  and let $\restact, \mathsf{c}$ (for \textbf{r}estart and \textbf{c}ontinue) be two actions. We associate to a $\Sigma$-observable Markov chain $\Mc = (\St, \init, \Sigma, \Lab, \Pm)$ a \emph{restart \acron{MDP}} $\restartMDP = (\St,  \init, \Lab, \{ \restact, \mathsf{c}\}, \transitions)$, where for every two states $s, s' \in \St$ the transition function is given by: $\transitions(s, \restact, s')=1$ if $s' = s_{in}$ and $0$ otherwise,  and $\transitions(s, \bm{c}, s')= \Pm(s, s')$.  Intuitively, at every state of $\restartMDP$ we have the choice between restarting the chain $\Mc$ or continuing. 

We consider black-box strategies on $\Sigma$ and $\{\restact, \mathsf{c}\}$.
Observe that if a run $\pi$ of $\restartMDPst{\strategy}$ contains finitely many occurrences of $\restact$, then the suffix of $\pi$ after the last occurrence of $\restact$ is a run of $\Mc$ (after dropping the occurrences of the continue action $c$). More precisely, if $\pi = \pi_0\pi'$, where $\pi'$ is the longest suffix of $\pi$ not containing $\restact$, then $\pi'=(\pi_0 \, s_{in}) \, (\pi_0 \, s_{in} \, \mathsf{c} \, s_1) \, (\pi_0 \, s_{in} \, \mathsf{c} \, s_1 \, \mathsf{c} \, s_2) \ldots$, where $s_{in} s_1 s_2 \ldots$ is a run of $\Mc$. The sequence of observations of $s_{in} s_1 s_2 \ldots$ is an infinite word over $\Sigma$, called the \emph{tail} of $\pi$; formally $\tail{\pi}\coloneqq  \Lab(s_{in}) \Lab(s_1) \Lab(s_2) \cdots$. 

\begin{definition}[Black-box testing strategies]
Let $\lang \subseteq \Sigma^\omega$ be an $\omega$-regular language. A \emph{black-box strategy} $\strategy$ on $\Sigma$ and  $\{\restact, \mathsf{c}\}$ is a \emph{testing strategy for $\lang$} if it satisfies the following property:  for every $\Sigma$-observable Markov chain $\Mc$,  if  $\Pr_\Mc(\lang) > 0$ then w.p.1 a run of $\restartMDPst{\strategy}$ has a finite number of restarts, and its tail belongs to $L$. The \emph{black-box testing problem} for $\lang$ consists of finding a black-box testing strategy for $\lang$. 
\end{definition}

We denote by $\restart(\rho)\in\N\cup\{\infty\}$ the number of appearances of the restart action $\restact$ in $\rho$. Intuitively, the language $\lang$ models the set of potential violations of a given liveness specifications. If we sample any finite-state $\Sigma$-observable Markov chain $\Mc$ according to a testing strategy for $\lang$, then w.p.1 we eventually stop restarting, and the tail of the run is a violation, or there exist no violations.

\subsection{Canonical black-box testing problems}

Using standard automata-theoretic techniques, the black-box testing problem for an arbitrary $\omega$-regular language $\lang$ can be reduced to the black-box testing problem for a canonical language.   
For this, we  need to introduce some standard notions of the theory of automata on infinite words. 

A deterministic Rabin automaton (\DRA{}) over an alphabet $\Sigma$ is a tuple $\dra = (\draS, \Sigma, \draTr, \draInit, \draAcc)$, where $\draS$ is a finite set of states, $\draTr \colon \draS \times \Sigma \to \draS$ is a transition function, $\draInit \in \draS$ is the initial state, and $\draAcc \subseteq 2^\draS \times 2^\draS$ is the acceptance condition. The elements of $\draAcc$ are called \emph{Rabin pairs}.  A word $w = a_0a_1a_2 \ldots \in \Sigma^\omega$ is accepted by $\dra$ if the unique run $\draInit q_1 q_2 \ldots$ of $\dra$ on $w$ satisfies the following condition: there exists a Rabin pair $(E,F) \in \draAcc$ such that $a_i \in E$ for infinitely many $i \in \N$ and $a_i \in F$ for finitely many $i \in \N$. It is well known that \DRA{}s recognize exactly the $\omega$-regular languages (see e.g.~\cite{Baier2008}). The \emph{Rabin index} of an $\omega$-regular language $L$ is the minimal number of Rabin pairs of the \DRA{}s that recognize $L$.

\newcommand{\infl}{\bm{e}}
\newcommand{\finl}{\bm{f}}

\begin{definition}
Let $k \geq 1$, and let $M_k = \{\infl_1, \ldots, \infl_k, \finl_1, \ldots, \finl_k\}$ be a set of \emph{markers}. The \emph{Rabin language} $\rabin_k \subseteq (2^{M_k})^\omega$ is the  language of all words $w = \alpha_0\alpha_1 \cdots  \in (2^{M_k}) ^\omega$ satisfying the following property: there exists $1 \leq j \leq k$ such that $\infl_j\in \alpha_i$ for infinitely many $i \geq 0$, and $\finl_j \in \alpha_i$ for at most finitely many $i \geq 0$. 
\end{definition}

We show that the black-box testing problem for languages of Rabin index $k$ can be reduced to the black-box testing problem for $\rabin_k$. 

\begin{restatable}{lemma}{rabinlemma}
\label{lem:rabin}
There is an algorithm that, given an $\omega$-regular language $\lang \subseteq \Sigma^\omega$ of index $k$ and given a testing strategy $\strategy_k$ for $\rabin_k$, effectively constructs a testing strategy $\strategy_\lang$ for $\lang$. 
\end{restatable}
\begin{proof}
(Sketch, full proof in the Appendix.) Let $\dra = (\draS, \Sigma, \draTr, \draInit, \draAcc)$ be a \DRA{} recognizing $\lang \subseteq \Sigma^\omega$ with accepting condition $\draAcc=\{(E_1, F_1), \ldots, (E_k, F_k)\}$, i.e., $\draAcc$  contains $k$ Rabin pairs. Let $\strategy_k$ be a black-box strategy for the Rabin language $\rabin_k$.  We construct a black-box strategy $\strategy_\lang$ for $\lang$. 

Let $w = \ell_1 a_1 \ell_2 \cdots \ell_{n-1} a_n \ell_n \in (\Sigma \times \{\restact, \mathsf{c}\})^* \Sigma$. We define the action $\strategy_\lang(w)$ as follows. Let $\draInit q_1 \ldots q_n$ be the unique run of $\dra$ on the word $\ell_1\ell_2 \ldots \ell_n \in \Sigma^*$. Define $v = \ell_1' a_1 \ell_2' \cdots \ell_{n-1}' a_n \ell_n' \in (2^{M_k} \times \{\restact, \mathsf{c}\})^* 2^{M_k}$ as the word given by: $\infl_j \in \ell_i'$ if{}f $q_i \in E_j$, and  $\finl_j \in \ell_i'$ if{}f $q_i \in F_j$. (Intuitively, we mark with $\infl_j$ the positions in the run at which the \DRA{} visits $E_j$, and with $\finl_j$ the positions at which the \DRA{} visits $F_j$.) We set $\strategy_\lang(w) \coloneqq  \strategy_k(v)$. We show in the Appendix that $\strategy_\lang$ is a black-box strategy for $\lang$. 
\end{proof}

\section{Black-box strategies for Rabin languages}
\label{sec:strategy}
\newcommand{\secondhalf}[1]{\textit{SecondHalf}({#1})}
\newcommand{\last}[1]{\textit{last}(#1)}

We describe a family of testing strategies for the Rabin languages $\{\rabin_k \mid k \geq 1 \}$. 
 In Section \ref{subsec:strategy} we describe our strategy in detail. In Section \ref{subsec:progress} we introduce the progress radius and the
 progress probability, two parameters of a chain needed to prove correctness and necessary for quantitative analysis in Section \ref{sec:quantitative}. In Section \ref{subsec:correctness} we formally prove that our strategy works.

\subsection{The strategy}
\label{subsec:strategy}
Let $\Mc$ be a Markov chain with observations in $2^{M_k}$, and let $\mypath = s_0 s_1 s_2 \cdots s_m$ be  a finite path of $\Mc$.  The \emph{length} of $\mypath$ is $m$, and its last state, denoted $\last{\pi}$, is $s_m$. 
The \emph{second half} of $\mypath$ is the path $\secondhalf{\pi}\coloneqq  s_{\lceil m / 2 \rceil} \ldots s_m$.
The \emph{concatenation} of $\mypath$ and a finite path $\rho= r_0 r_1 \cdots r_l$ of $\Mc$ such that $s_m= r_0$ is the path $\mypath \odot \rho \coloneqq  s_0 s_1 s_2 \cdots s_m r_1 \cdots r_l$. A path
 $\mypath$ is \emph{$i$-good} if it has length $0$ or there are markers $\infl_i, \finl_i\in M_k$ such that some state $s$ of $\mypath$ satisfies $\infl_i \in \Lab(s)$ and no state $s$ of $\mypath$ satisfies $\finl_i \in \Lab(s)$. 
 Further, $\mypath$ is \emph{good} if it is $i$-good for some $1 \leq i \leq k$.

The strategy $\mathfrak{S}[f]$, described in  Figure \ref{fig:strat}, is parametrized by a function $f \colon \N \rightarrow \N$. The only requirement on $f$ is $\limsup_{n\to \infty} f(n)=\infty$.
\begin{figure}[h]
\begin{center}
\begin{minipage}{11.0cm}
\begin{algorithmic}
	\State{$n\coloneqq 0$} \Comment{ number of restarts }
	\While{true}
	\State{$\mypath \gets \init$} \Comment{ initial state of the chain}
	\While{$\secondhalf{\mypath}$ is good}
	\State{sample path $\rho$ from state $\last{\pi}$}
	\State{of length $2 \cdot f(n)$}  \Comment{even length for convenience}
	\State{$\mypath \gets \mypath \odot \rho$}
	\EndWhile \Comment{ restart}
	\State{$n \gets n+1$}
	\EndWhile
\end{algorithmic}
\end{minipage}
\end{center}
\caption{Strategy $\mathfrak{S}[f]$ for the Rabin language $\rabin_k$ and a function $f \colon \N \rightarrow \N$.}\label{alg:cstrategy}
\label{fig:strat}
\end{figure}
In words, after the $n$-th restart the strategy keeps sampling in blocks of $2 \cdot f(n)$ steps until
the \emph{second half} of the complete path sampled so far is bad, in which case it restarts. For example, 
after the $n$-th restart the strategy samples a block $\mypath_0 = \mypath_{01} \odot \mypath_{02}$, where $|\mypath_{01}|=|\mypath_{02}|= f(n)$, and checks whether $\mypath_{02}$ is good; if not, it restarts, otherwise it samples a block $\mypath_1 = \mypath_{11} \odot \mypath_{12}$ starting from $\last{\mypath_{02}}$, and checks whether $\mypath_{11} \odot \mypath_{12}$ is good; if not, it restarts; if so it samples a block $\mypath_2 = \mypath_{21} \odot \mypath_{22}$ starting from $\last{\mypath_{12}}$, and checks whether $\mypath_{12} \odot \mypath_{21} \odot \mypath_{22}$ is good, etc. 
Intuitively, the growth of $f$ controls how the strategy prioritizes deep runs into the chain over quick restarts while the number of restarts increases.

In the rest of the paper we prove that our strategy is correct, and obtain optimal upper bounds on the number of steps to the last reset. These bounds are given in terms of two parameters of the chain: the progress radius and the progress probability. We introduce the parameters in section \ref{subsec:progress}.

\subsection{Progress radius and progress probability}
\begin{figure}
	\resizebox{1.09\linewidth}{!}{
	{
		\fontsize{7pt}{9pt}\selectfont
		\def\svgwidth{5.333in}
		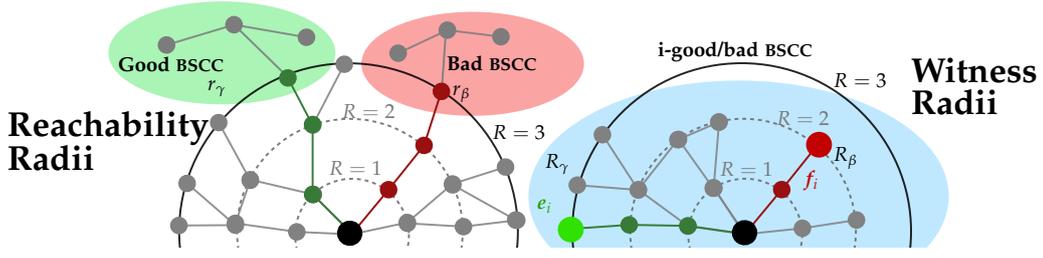
	}	
	}
	\caption{\textit{Left:} Intuitively, $r_\gamma$ and $r_\beta$ denote an upper bound of how hard it is to reach a \BSCC{}. \textit{Right:} $R_\gamma$ and $R_\beta$ measure how hard it is to reach a state with label $\infl_i$/$\finl_i$ inside the \BSCC{}s}
	\label{fig:rsandps}
\end{figure}
\label{subsec:progress}
We define the notion of progress radius and progress probability for a Markov chain $\Mc$  with $2^{M_k}$ as set of observations and such that $\Pr(\rabin_k) > 0$. 
Intuitively, the progress radius  is the smallest number of steps such that, for any state of the chain, conducting only this number of steps one can 
``make progress'' toward producing a good run or a bad run. The progress probability gives a lower bound for the probability of the paths that make progress. 

We define the notions only for the case $k=1$, which already contains all the important features. The definition for arbitrary $k$ is more technical, and is given in Appendix \ref{appendix:progress}. 

\medskip\noindent\textbf{Good runs and good \BSCC{}s.} We extend the definition of good paths to good runs and good \BSCC{}s of a Markov chain.
A run $\run=s_0s_1s_2\dots$ is \emph{good} if $\infl_1$ appears infinitely often in $\run$ and $\finl_1$ finitely often, and \emph{bad} otherwise. 
So a run $\run$ is good if{}f there exists a decomposition of $\run$  into an infinite concatenation $\run \coloneqq  \mypath_0 \odot \mypath_1 \odot \mypath_2 \odot \cdots$ of non-empty paths such that $\mypath_1, \mypath_2, \ldots$ are good. We let $\Pg$ denote the probability of the good runs of $\Mc$.

A \BSCC{} of $\Mc$ is \emph{good} if it contains at least one state labeled by $\infl_1$ and no state labeled by $\finl_1$, and bad otherwise. 
It is well-known that the runs of any finite-state Markov chain reach a \BSCC{} and visit all its states infinitely often w.p.1 \cite[Thm. 10.27]{Baier2008}. 
It follows that good (resp. bad) runs eventually reach a good (resp. bad) \BSCC{} w.p.1. 

\medskip\noindent\textbf{Progress radius.} Intuitively, the progress radius $\Rm$ is the smallest number of steps such that, for any state $s$, by conducting $\Rm$ steps one can 
``make progress'' toward producing a good run\,---\,by reaching a good \BSCC{} or, if already in one, by reaching a state with observation $\infl_1$\,---\,or a bad run.

\begin{definition}[Good-reachability and good-witness radii]
Let $B_\gamma$ be the set of states of $\Mc$ that belong to good \BSCC{}s and let $\St_\gamma$ be the set of states from which it is possible to reach $B_\gamma$,
and let $s \in \St_\gamma$. A non-empty path $\pi$ starting at $s$ is a \emph{good progress path} if 
\begin{itemize}
\item $s \in \St_\gamma \setminus B_\gamma$, and $\pi$ ends at a state of $B_\gamma$; or
\item $s \in B_\gamma$, and $\pi$ ends at a state with observation $\infl_1$.  
\end{itemize}
The \emph{good-reachability radius} $r_\gamma$ is the maximum, taken over every $s \in \St_\gamma \setminus B_\gamma$, of the length of a shortest progress path for $s$.
The \emph{good-witness radius} $R_\gamma$ is the same maximum, but taken over every $s \in B_\gamma$.
\end{definition}

The bad-reachability and bad-witness radii, denoted $r_\beta$ and $R_\beta$ are defined analogously. Only the notion of progress path of a state$s \in B_\beta$ 
needs to be adapted. Loosely speaking, a bad \BSCC{} either contains no states with observation $\infl_1$, or it contains some
state with observation $\finl_1$.  Accordingly, if no state of the \BSCC{} of $s$ has observation $\infl_1$, then any non-empty path starting at $s$ is a progress path,
and otherwise a progress path of $s$ is a non-empty path starting at $s$ and  ending at a state with observation $\finl_1$.
We illustrate the definition of the reachability and witness radii in Figure \ref{fig:rsandps}.
We leave $r_\beta$, $R_\beta$, $p_\beta$, and $P_\beta$ undefined if the chain does not contain a bad \BSCC{}, and hence runs are good w.p.1.

\begin{definition}[Progress radius]
The \emph{progress radius} $\Rm$ of $\Mc$ is the maximum of $r_\gamma$, $R_\gamma$, $r_\beta$, and $R_\beta$. 
\end{definition}

\medskip\noindent\textbf{Progress probability.} From any state of the Markov chain it is possible to ``make progress'' by executing a progress path of length $\Rm$. 
However, the probability of such paths varies from state to state. Intuitively, the progress probability gives a lower bound on the probability
of making progress. 

\begin{definition}
Let $B_\gamma$ be the set of states of $\Mc$ that belong to good \BSCC{}s, let $\St_\gamma$ be the set of states from which it is possible to reach $B_\gamma$,
and let $s \in \St_\gamma$. The \emph{good-reachability probability} $p_\gamma$ is the minimum, taken over every $s \in \St_\gamma\setminus B_\gamma$, of the probability 
that a path with length $r_\gamma$ starting at $s$ contains a good progress path. The \emph{good-witness probability} $P_\gamma$ is the same mininum, but taken over every $s \in B_\gamma$ with paths of length $R_\gamma$.
The corresponding bad probabilities are defined analogously. The \emph{progress probability} $\Pmax$ is the  minimum of $p_\gamma, P_\gamma, p_\beta, P_\beta$.
\end{definition}

\newcommand{\NB}{\textit{NB}}
\subsection{Correctness proof}
\label{subsec:correctness}
We prove that the strategy $\mathfrak{S}[f]$ of section \ref{subsec:strategy} is a valid testing strategy $\mathfrak{S}[f]$ for arbitrary Markov chains $\mathcal{M}$.
First, we will give an upper bound on the probability that $\mathfrak{S}[f]$ restarts ``incorrectly'', i.e.~at a state $s\in S_\gamma$ from which a good \BSCC{} could still be reached.
\begin{restatable}{lemma}{ProbRestartGood}
\label{lemma:PRestartGood}
Let $\Mc$ be a Markov chain, and let $\restartMDPst{\mathfrak{S}[f]}$ be its associated Markov chain with 
$\mathfrak{S}[f]$ as restart strategy.  Let $\NB_{n}$ be the set of paths of $\restartMDPst{\mathfrak{S}[f]}$ that have at least $n-1$ restarts and only visit states in $S_\gamma$ after the $(n-1)$-th restart.  We have:
\[
\Pr[\restart \ge n \mid \NB_n]\le  3(1-\Pmax )^{\lfloor f(n)/\Rm \rfloor-1}
\]
\end{restatable}

The technical proof of this lemma is in the Appendix. We give here the proof for a special case that illustrates most ideas.
Consider the Markov chain with labels in $2^{ \{\infl_1, \finl_1\} }$ at the top of Figure \ref{fig:examplemcrun}. 
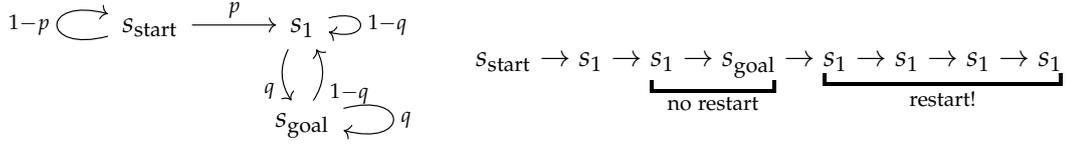
\begin{figure}
\begin{minipage}{6cm}
\begin{equation*}
	\begin{tikzcd}
  s_{\text{start}}\ar[loop left, "1-p"]\ar[r,"p"]& s_1\ar[loop right,"1-q"]\ar[d,"q"', bend right]\\
  & s_\text{goal}\ar[loop right,"q"]\ar[u,bend right,"1-q"']
	\end{tikzcd}
\end{equation*}
\end{minipage}
\quad
\begin{minipage}{6cm}
\begin{equation*}
 		s_{\text{start}}\rightarrow s_{1}\rightarrow \underbracket{s_1 \rightarrow s_\text{goal}}_\text{no restart} \rightarrow \underbracket{s_1 \rightarrow s_1 \rightarrow s_1 \rightarrow s_1}_\text{restart!}
 	\end{equation*}
 	\end{minipage}
 \caption{A Markov chain (left), and (a finite prefix of) one of its runs for $n=2$ and $f(n)=n$ (right). After $2$ steps, the run has reached a good \BSCC{} $\mathcal{B}$. After $4$ steps, $\strategy$ checks whether to restart, but decides against it because of $s_\text{goal}$ in step $4$. After 8 steps it checks again, restarting this time. This restart is covered by the third case of the case distinction, with $k=4$. The run must have visited $s_\text{goal}$ in step 4, because otherwise the minimal $k$ would be $3$ or less.}
 \label{fig:examplemcrun}
 \end{figure}
The labeling function is
$\Lab(s_\text{goal}) = \{ \infl_1 \}$ and $\Lab(s) = \emptyset$ for all other states, and $\Lab(\run) \in \rabin_1$ if{}f $\rho$ visits
$s_\text{goal}$ infinitely often.
The set $S_\gamma$ contains all states because $s_\text{goal}$ is reachable from every state. The only \BSCC{} is $\mathcal{B} = \{s_1,s_\text{goal}\}$, and it is a good \BSCC{}. 
From the definitions of the parameters we obtain $r_\gamma = R_\gamma = 1$, $p_\gamma = p$ and $P_\gamma = q$. Further, since there are no bad \BSCC{}s, $r_\beta$ and $R_\beta$ are undefined, and so $\Rm  = 1$. 
So for this Markov chain Lemma \ref{lemma:PRestartGood} states $\Pr[\restart \ge n \mid \NB_n]\le  3(1-\Pmax )^{f(n)-1}$. Let us see why this is the case.

Let $\run$ be a run of $\restartMDPst{\mathfrak{S}[f]}$ such that $\restart(\run) \geq n$, i.e., $\run$ has at least $n$ restarts. Since $S_\gamma$ contains all states, we have $\run \in \NB_n$ if{}f $\restart(\run) \geq n-1$. We consider three cases. In the definition of the cases we start counting steps immediately after the $(n-1)$-th restart, and denote by $\run[a, b]$ the fragment of $\run$ that starts immediately before step $a$, and ends immediately after step $b$.
\begin{itemize}
	\item[(a)] After $f(n)$ steps, $\run$ has not yet reached $\mathcal{B}$. \\
	Then $\run$ has stayed in $s_\text{start}$ for $f(n)$ consecutive steps, which, since $p=p_\gamma$, happens with probability at most $(1-p_\gamma)^{f(n)}$.
	\item[(b)] After $f(n)$ steps, $\run$ has already reached $\mathcal{B}$. Further, the $n$-th restart happens 
		immediately after step $2 f(n)$.\\
	In this case, by the definition of the strategy, $\run$ does not visit $s_\text{goal}$ during the interval $\run[f(n)+1, 2f(n)]$ (the second half of $[0, 2f(n)]$). So  $\run$ stays in $s_1$ during the interval $\run[f(n)+1, 2f(n)]$ which, since $\run$ has already reached $\mathcal{B}$ by step $f(n)$, occurs with probability $(1 -P_\gamma)^{f(n)}$.
	\item[(c)] After $f(n)$ steps, $\run$ has already reached $\mathcal{B}$. Further, the $n$-th restart does not happen before step $2f(n)+1$. \\ 
	Since the $n$-th restart happens at some point, and not before step $2f(n)+1$, by the definition of the strategy there is a smallest 
	number $k \geq f(n)$ such that $\run$ does not visit $s_\text{goal}$ during the interval $\run[k+1, 2k]$. Because we assume that the $n$-th restart happens after step $2f(n)$, we even have $k>f(n)$. 
	By the minimality of $k$, the run $\run$ does visit $s_\text{goal}$ during the interval $\run[k, 2k-2]$. 
	So $\run$ moves to $s_{goal}$ at step $k$, and then stays in $s_1$ for $k$ steps. The probability of the runs that eventually move to $s_{goal}$ and then move to stay in $s_1$ for $k$ steps is $P_\gamma(1-P_\gamma)^k$.
\end{itemize}
Figure \ref{fig:examplemcrun} shows at the bottom an example of a run, and how the stratgy handles it.
Since (a)-(c) are mutually exclusive events,  $\Pr[\restart \ge n \mid \NB_n]$ is bounded by the sum of their probabilities, where in case (c) we sum over all possible values of $k$. This yields:
\begin{equation*}
Pr[\restart \ge n \mid \NB_n] \le \;  (1-p_\gamma)^{f(n)}+(1-P_\gamma)^{f(n)}
+ \sum_{k=f(n)+1}^\infty P_\gamma(1-P_\gamma)^k\le \;  3(1-\Pmax )^{f(n)}.
\end{equation*}
The proof for arbitrary Markov chains given in the Appendix has the same structure, and in particular the same split into three different events. 
Applying Lemma \ref{lemma:PRestartGood} we now easily obtain (see the Appendix for a detailed proof) an upper bound for the probability to restart an $n$-th time. Note that this bound captures the ``correct'' as well as the ``incorrect'' restarts:

\begin{restatable}[Restarting probability]{lemma}{ProbRestart}
	\label{lemma:restartingProb}
Let $\Mc$ be a Markov chain, and let $\restartMDPst{\mathfrak{S}[f]}$ be its associated Markov chain with 
$\mathfrak{S}[f]$ as restart strategy. The probability that a run restarts again after $n-1$ restarts satisfies:
	\[
	\Pr [\restart\ge n\mid \restart\ge n-1]\le 1- \Pg \left( 1 - 3(1-\Pmax )^{\lfloor f(n)/\Rm \rfloor-1} \right) 
	\]
\end{restatable}
\begin{proof}
	Let $\NB_{n}$ be the set of paths of $\restartMDPst{\mathfrak{S}[f]}$ that have at least $n-1$ restarts and only visit states in $S_\gamma$ after the $(n-1)$-th restart and $\overline{\NB}_n$ its complement. We have 
	\begin{align*}
		\Pr [\restart\ge n\mid \restart\ge n-1] = & \Pr[\restart \ge n \mid \NB_{n}]  \cdot \Pr[\NB_{n} \mid \restart \ge n-1] + \\
		& \Pr[\restart \ge n \mid \overline{\NB}_n] \cdot \Pr[\overline{\NB}_n \mid \restart \ge n-1].
	\end{align*}
	Applying  Lemma \ref{lemma:PRestartGood} and $\Pr[\restart \ge n \mid \overline{\NB}_n]\le 1$, we get
\begin{align*}
                    \Pr [\restart\ge n\mid \restart\ge n-1]  \leq &\;  \big( 3(1-\Pmax )^{\lfloor f(n)/\Rm \rfloor -1} \big)  \Pr[\NB_{n} \mid \restart \ge n-1]\\
                     &+  \Pr[\overline{\NB}_n \mid \restart \ge n-1]	  
\end{align*}
W.p.1, good runs of $\Mc$ only visit states of $S_\gamma$ and so $\Pr[\NB_{n} \mid \restart \ge n-1] \ge \Pg$ and thus $\Pr[\overline{\NB}_{n} \mid \restart \ge n-1] \le 1-\Pg$, which completes the proof.
\end{proof}

Finally, we show that $\mathfrak{S}[f]$ is a correct testing strategy. Further, we show that the condition $\limsup_{n\to \infty} f(n)=\infty$ is not ony sufficient, but also necessary.
The previous lemma gives an upper bound on the probability for a restart that, for increasing $f(n)$, drops below $1$. If $f(n)$ is above that threshold for infinitely many $n$, 
it suffices to show that the strategy $\mathfrak{S}[f]$ restarts every bad run:

\begin{theorem}
	\label{theorem:testingstrategies}
$\mathfrak{S}[f]$  is a testing strategy for the Rabin language $\rabin_k$ if{}f the function $f$ satisfies $\limsup_{n\to \infty} f(n)=\infty$. 
\end{theorem}
\begin{proof}
\noindent ($\Rightarrow$): We prove the contrapositive. If $\limsup_{n\to \infty} f(n) < \infty$ then there is a bound $b$ such that $f(n) \leq b$ for every $n \geq 0$. Consider a Markov chain over $2^{M_1}$ consisting of a path of $2b+1$ states, with the last state leading to itself with probability 1; the last state is labeled with $\infl_1$, and no state is labeled with $\finl_1$ . Then the chain has a unique run that goes from the initial to the last state of the path and stays there forever, and its observation is a word of $\rabin_1$; therefore, $\Pr(\rabin_1)=1$. However, since  $2 f(n) \leq 2b +1$, $\mathfrak{S}[f]$ always restarts the chain before reaching the last state.

\noindent ($\Leftarrow$): By the previous lemma, we can bound the restart probability after $n-1$ restarts by $\smash{1-\left(1-3(1-\Pmax )^{\lfloor f(n)/\Rm \rfloor-1}\right)\Pg}$. 
Because $0<\Pmax \le 1$ and and $\Pg>0$, for large enough $f(n)$ this is smaller than $1-\Pg/2<1$. Because of $\limsup_{n\to \infty} f(n)=\infty$, we have that the probability to restart the run another time is at most $1-\Pg/2$ for infinitely many $n$, and hence the total number of restarts is finite with probability $1$. A bad run would enter a bad \BSCC{} $B$ w.p.1 and would then go on to visit a set consisting of all the $f_i$ corresponding to $B$ infinitely often. Thus, $\mathfrak{S}[f]$ would restart this run and hence reached a good run when it does not restart.
\end{proof}

\section{Quantitative analysis}
\label{sec:quantitative}

The quality of a testing strategy is given by the expected number of steps until the last restart, because this is the overhead spent until a violation starts to be executed.  As in \cite{EKKW21}, given a labeled Markov chain $\Mc$ and a testing strategy $\strategy$, we define the number of steps to the last restart as random variables over the Markov chain $\restartMDPst{\strategy}$:

\begin{definition}[$S(\rho)$ and $S_n(\run)$]
Let  $\run$ be a run of $\restartMDPst{\strategy}$. We define: $S(\run)$ is equal to $0$ if $\restact$ does not occur in $\run$; it is equal to the length of the longest prefix of $\run$ ending in $\restact$, if $\restact$ occurs at least once and finitely often in $\run$; and it is equal to $\infty$ otherwise. Further, for every $n \geq 1$ we define $S_n(\run)$ to be equal to $0$ if $\restact$ occurs less than $n$ times in $\rho$; and equal to the length of the segment between the $(n-1)$-th (or the beginning of $\rho$) and the $n$-th occurrence of $\restact$.
\end{definition}

In this section we investigate the dependence of $\expected[S]$ on the function $f(n)$. A priori it is unclear whether $f(n)$ should grow fast or slow. Consider the case in which all \BSCC{}s of the chain, good or bad, have size 1, and a run eventually reaches a good \BSCC{} with probability $p$. In this case the strategy restarts the chain until a sample reaches a good \BSCC{} for the first time. If $f(n)$ grows fast, then after a few restarts, say $r$, every subsequent run reaches a \BSCC{} of the chain with large probability, and so the expected number of restarts is small, at most $r + (1/p)$. However, the number of steps executed during these few restarts is large, because $f(n)$ grows fast; indeed, only the run after the penultimate restart executes already at least $2f(r + (1/p) -1)$ steps.

In a first step we show that $\expected[S]= \infty$ holds for every function $f(n) \in 2^{\Omega(n)}$. 

\begin{proposition}
Let $f \in 2^{\Omega(n)}$. Then there exists a Markov chain such that the testing strategy of Figure \ref{fig:strat} satisfies $\expected(S)=\infty$.
\end{proposition}
\begin{proof}
Let $f$ be in $2^{\Omega(n)}$. Then there exists some positive integer $k>0$ such that we have
$\limsup_{n\to \infty} f(n)\cdot (1/2)^{n/k}>0.$
Consider a Markov chain with $\Pg=1-(1/2)^{1/k}$. Then we have $
\expected(S)=\sumi_{n=0}\expected (S_n\mid \restart\ge n-1)P(\restart\ge n-1).$
We have that $P(\restart\ge n-1\mid \restart\ge n-2)\ge 1-\Pg$ because only good runs will not be restarted. We also have that $\expected (S_n\mid \restart\ge n-1)\ge f(n)(1-\Pg)$ because of the same reason. Thus
\begin{align*}
	\expected[S]=&\sumi_{n=0}\expected (S_n\mid \restart\ge n-1)P(\restart\ge n-1)
	\ge \sumi_{n=0} f(n)(1-\Pg)\cdot (1-\Pg)^n
\end{align*}
and hence $\expected[S]\ge \displaystyle\sumi_{n=0} f(n)(1/2)^{n/k}=\infty$.
\end{proof}
It follows that (if we limit ourselves to monotonic functions, which is no restriction in practice), we only need to consider functions $f(n)$ satisfying $f(n) \in \omega(1) \cap 2^{o(n)}$. In the rest of the section we study the strategies corresponding to polynomial functions $f(n) = n^c$ for $c \in \N_+$, and obtain an upper bound as a function of the parameters $\Rm/\Pmax $, $\Pg$, and $c$. The study of subexponential but superpolynomial functions is beyond the scope of this paper.


\subsection{Quantitative analysis of strategies with $f(n) = n^c$}

We give an upper bound on $\expected (S)$, the expected total number of steps before the last restart. Our starting point is Lemma \ref{lemma:restartingProb}, which bounds the probability to restart for the $n$-th time, if $(n-1)$ restarts have already happened.  When the number $n$ of restarts is small, the value of the right-hand-side is above $1$, and so the bound is not useful. We first obtain a value $X$ such that after $X$ restarts the right-hand-side drops below $1$.

\begin{lemma}
	\label{lemma:ChoiceOfX}
	Let $X = \sqrt[c]{\Rm\left(2+\ln(1/6)/\ln(1-\Pmax )\right)}$.
	For all $n\ge X$, we have 
	\begin{equation*}
		\Pr [\restart\ge n\mid \restart\ge n-1]\le 1-\Pg/2
	\end{equation*}
when restarting according to $\mathfrak{S}[n \mapsto n^c]$.
\end{lemma}
\begin{proof}
Follows immediately from Lemma \ref{lemma:restartingProb}, the fact that the restart probability decreases with $n$, the definition of $X$, and some calculations. We recall the statement of Lemma \ref{lemma:restartingProb}:
\[
\Pr [\restart\ge n\mid \restart\ge n-1]\le 1- \Pg \left( 1 - 3(1-\Pmax )^{\lfloor f(n)/\Rm \rfloor-1} \right) 
\]
Plugging in an $n\ge X$ validates the claim.
\end{proof}

We now try to find a bound for $\expected[S]$: By linearity of expectation, we have $\expected[S]=\sum_{i=1}^\infty \expected[S_n]$. We split the sum into two parts: for $n<X$, and for $n\ge X$. For $n<X$ we just approximate $\Pr[\restart\ge n-1]$ by $1$. For $n>X$ we can say more thanks to Lemma \ref{lemma:ChoiceOfX}: 
	\begin{align*}
		\Pr[\restart\ge n-1]&=\Pr[\restart\ge n-1 | \restart\ge n-2] \cdots \Pr[\restart\ge X+1 | \restart\ge X]\cdot \Pr[\restart\ge X]\\
		&\le \prod_{k=\lceil X\rceil}^n \Pr[\restart\ge k | R\ge k-1] \le \left(1-\Pg/2\right)^{n-X}
	\end{align*}
	This yields:
	\begin{align*}
		\expected[S]&=\sum_{n=0}^\infty \expected[S_n\mid\restart \ge n-1]\Pr[\restart\ge n-1]\\ 
		&\le \sum_{n=0}^{X}\expected[S_n\mid\restart \ge n-1] + \sum_{n=X}^\infty \expected[S_n\mid\restart \ge n-1]\cdot \left(1-\Pg/2\right)^{n-X}  \tag{1}
	\end{align*}
It remains to bound the expected number of steps between two restarts $\expected[S_n\mid\restart \ge n-1]$, which is done in 
Lemma \ref{lemma:boundnsmall} below. The proof can be found in the Appendix. The proof first observes that the expected number of steps it takes to reach a good or a bad \BSCC{} is $r_\gamma/p_\gamma$ resp. $r_\beta/p_\beta$. Then we give a bound on the expected number of steps it takes to perform a progress path inside a bad \BSCC{} for the first time, or to not perform a progress path inside a good \BSCC{} for an entire second half of a run at some point after the $(n-1)$-st restart; the bound is also in terms of $\Rm/\Pmax $ and $\Rm/\Pmax (1 - P_\gamma)$. The term $2 f(n)$ comes from the fact that the strategy always executes at least $2 f(n)$ steps. The term $2 \Rm$ is an artifact due to the ``granularity'' of the analysis, where we divide runs in blocks of $\Rm$ steps.

\begin{restatable}[Expected number of steps in a fragment]{lemma}{stepsfragment}
	\label{lemma:boundnsmall}
	For the strategy $\mathfrak{S}[n \mapsto n^c]$ we have:
\begin{equation}
	\expected[S_n\mid\restart \ge n-1]\le 2 (\Rm+f(n))+9\left(\frac{\Rm}{\Pmax (1-P_\gamma)}\right).
\end{equation}
\end{restatable}

\noindent Plugging Lemma \ref{lemma:boundnsmall} into (1), we finally obtain (see the Appendix):

\begin{restatable}[Expected number of total steps]{theorem}{totalsteps}
	\label{theorem:totalsteps}
	For the strategy $\mathfrak{S}[n \mapsto n^c]$ we have:
	\begin{equation*}
		\Esteps \in \bigO \left((c+1)!\cdot 2^c\cdot \left(\frac{\Rm}{\Pmax }\right)^{1+1/c}+\frac{2^c(c+1)!}{P_{good}^{c+1}} + (c+1)!(2c)^{c+1}\right).
	\end{equation*}
\end{restatable}
If we fix a value $c$, we obtain a much simpler statement:
\begin{corollary}
	\label{cor:FinalBound}
	For a fixed $c$, the strategy $\mathfrak{S}[n \mapsto n^c]$ satisfies:
\begin{equation*}
	\Esteps \in \bigO \left(\left(\frac{\Rm}{\Pmax }\right)^{1+1/c}+\frac{1}{\Pg^{c+1}}\right).
\end{equation*}
\end{corollary}
Thus the bound on the total number of steps depends on two quantities, $\Rm/\Pmax$ and $\Pg$. A small $c$ favours the effect of $\Rm/\Pmax$ on the bound, a larger $c$ the effect of $\Pg$. 
In Section \ref{sec:experiments} we will see that this closely matches the performance of the algorithms for different values of $c$ on synthetic Markov chains and on Markov chains from the \acron{PRISM} benchmark set.
\subsection{Optimality of the strategy $f(n) = n^c$}
We will prove the following optimality guarantee for our strategies.
\begin{theorem}
	For every $c \in \N_+$ there is a family of Markov chains such that our bound of Corollary \ref{cor:FinalBound} on $\mathfrak{S}[n\mapsto n^c]$ is asymptotically optimal, i.e., no other black-box testing strategy is in a better asymptotic complexity class.
\end{theorem}
This proves two points: first, our bounds cannot be substantially improved.  Second, one necessarily needs information on  $\frac{\Rm}{\Pmax }$ and $\Pg$ to pick an optimal value for $c$; without any information every value is equally good.
\begin{proof}
Consider the family of Markov chains at the top of Figure \ref{MC:SyntheticPRm}. 
We take an arbitrary $k>1$ and set $M=k^{c-1}$ and $p=q=1/k$. With this choice we have
$\Pg = \Pmax  = 1/k$, and $\Rm=k^{c-1}$. By Lemma \ref{theorem:totalsteps}, 
the strategy $\mathfrak{S}[n\mapsto n^c]$. satisfies  $\expected[S] \in \bigO((\Rm/\Pmax )^{1+1/c}+ (1/\Pg)^{c+1})=\bigO (k^{c+1})$

We compare this with the optimal number of expected steps before the final restart. 
Since runs that visit $s_\text{goal}$ at least once are good w.p.1, any optimal strategy 
stops restarting exactly after the visit to  $s_\text{goal}$. We claim that every
such strategy satisfies $\expected[S] \ge \Rm/(\Pg \Pmax )(1-\Pg)$. For this, we make four observations. First, the probability of a good run is $\Pg$. Second, the expected number of steps of a good run until the first visit to $s_\text{goal}$ is $\Rm/ \Pmax $. Third, the smallest number of steps required to distinguish a bad run, i.e.~being in the left \BSCC{}, from a good run is equal to $\Rm$, because until $\Rm$ steps are executed, all states visited carry the same label. Hence, every strategy takes $\Rm/(\Pg \Pmax )$ steps on average before reaching the state $s_\text{goal}$ for the first time. Fourth, on average $1/\Pg$ tries are required to have one try result in a good run. Hence, on average at least $\frac{1/\Pg-1}{1/\Pg}$ of the $\Rm/(\Pg \Pmax )$ steps happen before the last restart. Since $\frac{1/\Pg-1}{1/\Pg}=(1-\Pg)$, this proves the claim. Now $\Rm/(\Pg \Pmax )(1-\Pg)=k^{c+1}-k^c\in \Theta(k^{c+1})$ and we are done. 
\end{proof}

\section{Experiments}
\label{sec:experiments}
We report on experiments on three kinds of systems. First, we conduct experiments  on two synthetic families of Markov Chains. Second, we repeat the experiments of  \cite{EKKW21} on models from the standard \acron{PRISM} Benchmark Suite \cite{DBLP:conf/qest/KwiatkowskaNP12} using our black-box strategies.  Finally, we conduct experiments on population protocols from the benchmark suite of the Peregrine tool \cite{BlondinEJ18,EsparzaHJM20}.

\medskip\noindent\textbf{Synthetic Experiments.}
\begin{figure}[t]
	\begin{equation*}
		\begin{tikzcd}
			& & & & \mathbf{s_\text{goal}}\ar[bend right,dll,"1"']& \\
			s_\text{sink} \ar[loop left]& s_{\text{start}}\ar[l, "1-q"']\ar[r,"q"]& s_1\ar[r,"1"]&\dots\ar[r,"1"]&s_{M-1}\ar[r,"1-p"]\ar[u,"p"'] &s_M\ar[bend left,lll,"1"]
		\end{tikzcd}
      \end{equation*}
      \begin{equation*}
\begin{tikzcd}
		s_\text{sink} \ar[loop left]& s_{\text{start}}\ar[l, "1-p"']\ar[r,"p"]&  s_{2}\ar[ll, bend left, "1-p"']\ar[r,"p"]&\dots  \ar[lll, bend left, "1-p"', in=140] \ar[r,"p"]&  s_{M}\ar[llll, bend left, "1-p"', in=130]\ar[r,"p"] &  s_\text{goal}\ar[loop right]
	\end{tikzcd}
\end{equation*}
	\caption{Two families of Markov chains. The initial state is $s_\text{start}$. The good runs are those that visit $s_\text{goal}$ infinitely often. For the top chains, $\Pg=q$, $\Rm  = M$, and  $\Pmax  = p$. For the bottom chains, $\Pg=p^M$, $\Rm =M$, $\Pmax =p^M$.}
	\label{MC:SyntheticPRm}
	\end{figure}
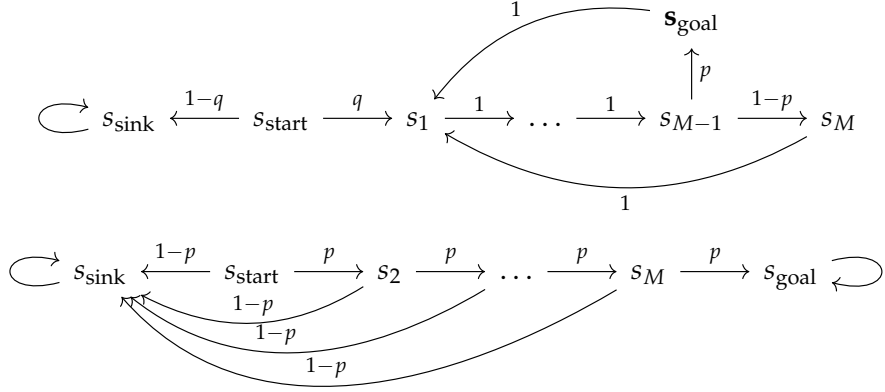
Consider the two (families of) labeled Markov chains at the top of Figure \ref{MC:SyntheticPRm}. The labels are
$a$ and $b$. In the top chain, state $s_{\textit{goal}}$ is labeled by $a$, all others by $b$. In the bottom chain, the states $s_2$ to $s_M$ and $s_\text{goal}$ are labeled by $\{a,s_\text{start}\}$ and $s_\text{sink}$ by $b$. The language $L$ is the set of words containing infinitely many occurrences of $a$.  In the top chain at the initial state we go right or left with probability $q$ and $(1-q)$, respectively. Runs that go left are bad, and runs that go right are good w.p.1. It follows $\Pg=q$, $\Rm  = M$, and  $\Pmax  = \min (p,q)$. In our experiments we fix $q = 1/2$. By controlling $M$ and $p$, we obtain chains with different values of $\Rm $ and $\Pmax $ for fixed $\Pg=1/2$. 
In the bottom chain, $R_\beta=R_\gamma =1$, $\Rm =r_\gamma=M$, $p_\gamma=p^M$, $p_\beta=(1-p)$ and $\Pmax =\min (p^M,1-p)$ and $\Pg=p^M$. 

\begin{figure}[h!]
	\centering
	\begin{subfigure}{.5\textwidth}
		\centering
		\resizebox{
			\linewidth}{!}{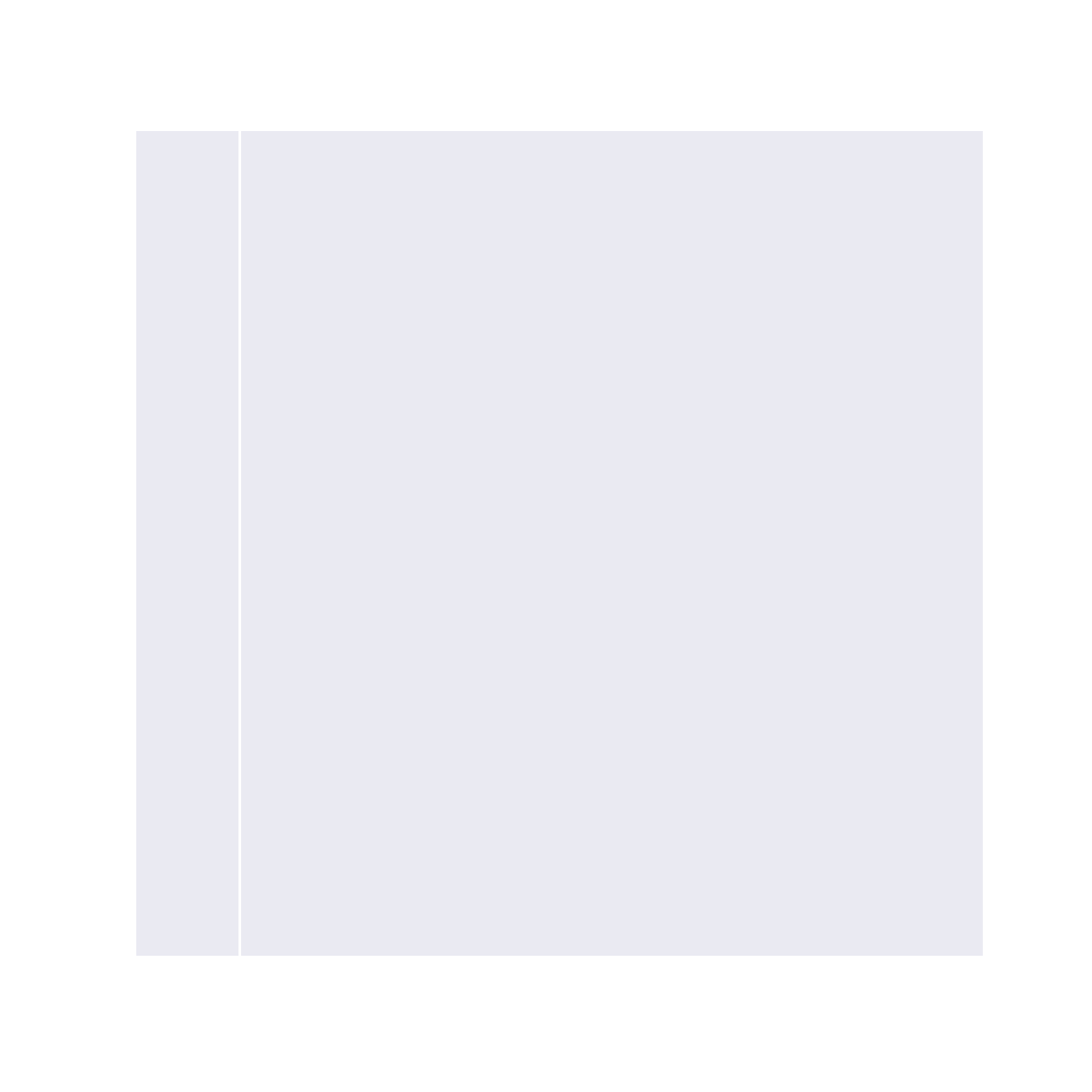}
		\label{subfig:SynthExpRp}
	\end{subfigure}%
	\begin{subfigure}{.5\textwidth}
		\centering
			\resizebox{
			\linewidth}{!}{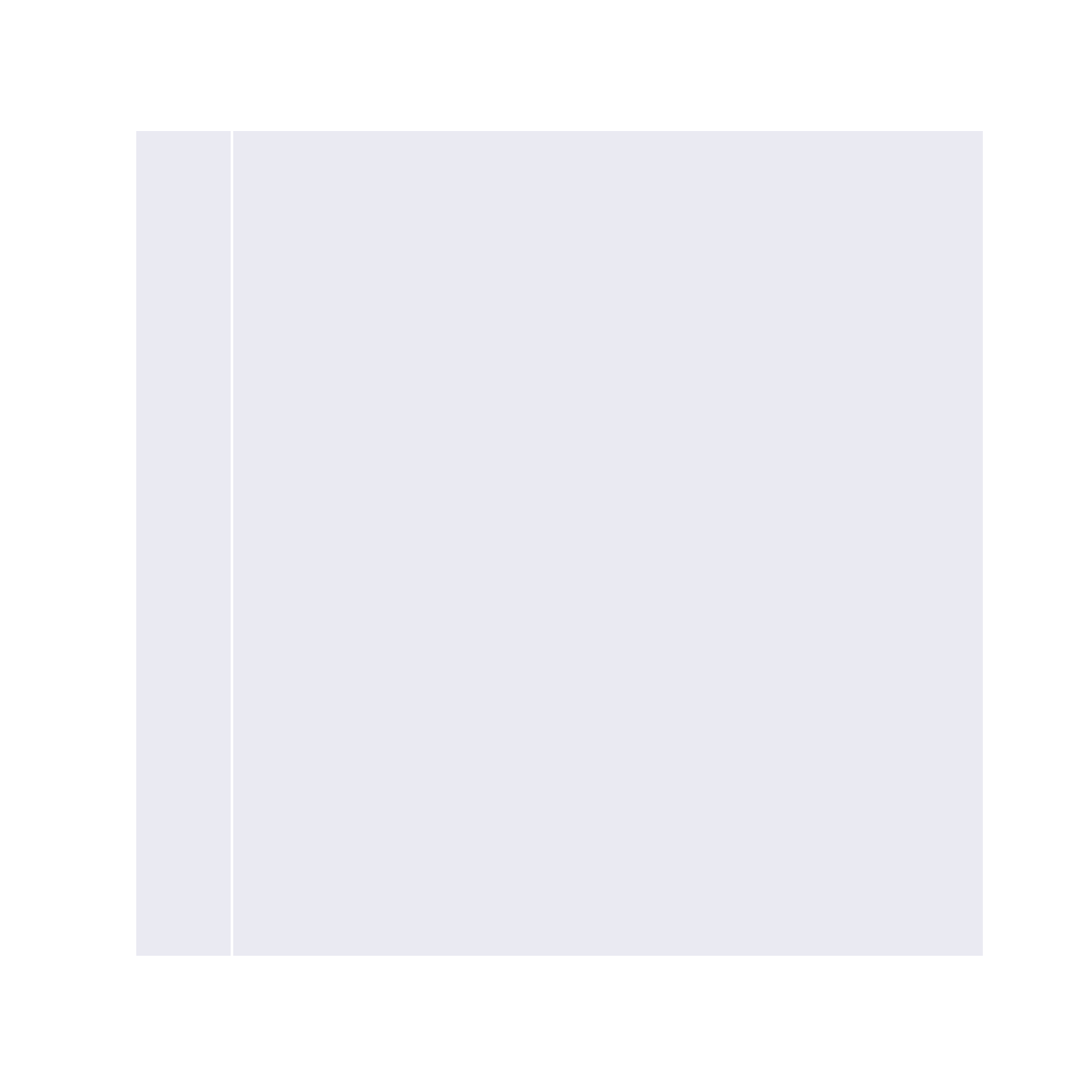}
		\label{subfig:SynthExpPgood}
	\end{subfigure}
	\caption{ 	
	On the left, double-logarithmic plot of the expected total number of steps before the last restart $\expected(S)$ for the chain at the top of Figure \ref{MC:SyntheticPRm} as a function of $\Rm /\Pmax$ for strategies (\ref{fig:strat}) with $f(n)=n^c$ for varying $c$. On the right, same for the bottom chain as a function of $1/\Pg$. The plots also show linear regressions. The leading exponent can be taken from the legend.}
	\label{fig:SynthExp}
\end{figure}
Recall that the bound obtained in the last section is $\Esteps \le f(c)(\Rm /\Pmax )^{1+1/c}+g(c) (1 / \Pg)^{c+1}$
where $f(c)$ and $g(c)$ are fast-growing functions of $c$.
If $\Pg$ and $\Rm /\Pmax $ are small, then $f(c)$ and $g(c)$ dominate the number of steps, and hence strategies with small $c$ should perform better. The data confirms this prediction.
Further, for fixed $\Pg$, the bound predicts $\Esteps \in O((\Rm /\Pmax )^{1+1/c})$, and so for growing $\Rm /\Pmax $ strategies with large $c$ should perform better. The left diagram confirms this. Also, the graphs become straight lines in the double logarithmic plot, confirming the predicted polynomial growth. 
Finally, for $\Rm /\Pmax $ and $1/\Pg$ growing roughly at the same speed as in the lower Markov chain, the bound predicts $\Esteps \in O(1 / \Pg^{c+1})$ for $c=2,3$ and $\Esteps \in O(M^2/ \Pg^{c+1})$ for $c=1$, and hence for growing $\Pg$ and $\Rm /\Pmax $, strategies with small $c$ perform better.  Again, the right diagram confirms this.

\medskip\noindent\textbf{Experiments on the \acron{PRISM} Data set.}
We evaluate the performance of our black-box testing strategies for different values of $c$ on discrete time Markov chain benchmarks from the \textsc{\acron{PRISM}} Benchmark suite \cite{DBLP:conf/qest/KwiatkowskaNP12}, and compare them with the strategies of  \cite{EKKW21} for fully observable systems. Table \ref{tableconcur} shows the results.
The properties checked are of the form $\mathbf{GF},  (\mathbf{GF}\rightarrow \mathbf{FG})$, or their negations. We add a gridworld example\footnote{Unfortunately, the experimental setup of \cite{EKKW21} cannot be applied to this example \cite{We22}.} denoted $\overline{\mathtt{GW}}$, with larger values of the parameters, to increase the number of states to $\sim 5\cdot 10^8$. When trying to construct the corresponding Markov chain, Storm experienced a timeout. Runs are sampled using the simulator of the \textsc{Storm} Model Checker \cite{DBLP:journals/corr/abs-2002-07080} and the python extension \textsc{Stormpy}. We abort a run after  $10^6$ (Up to $3\cdot 10^7$ for the gridworld examples $\mathtt{gw}$,  $\overline{\mathtt{gw}}$, and $\overline{\mathtt{GW}}$) steps without a restart. The probability of another restart is negligibly small. 

The Cautious\textsubscript{10}- and the Bold\textsubscript{$0.1$}-strategy of \cite{EKKW21}
store the complete sequence of states observed, and so need linear memory in the length of the sample. Our strategies use at most a logarithmic amount of memory, at none or little cost in the number of steps to the last restart. Our strategies never timeout and, surprisingly, often require \emph{fewer} steps than fully-observable ones. In particular, the strategies for fully observable systems cannot handle $\overline{\text{gridworlds}}$, and only the bold strategy handles gridworld. One reason for this difference is our strategies' ability to adapt to the size of the chain automatically by increasing values of $f(n)$ as $n$ grows. In two cases (nand and bluetooth) the fully observing strategies perform better by a factor of $\sim 2$ to $\sim 3$. 
In comparison to the improvement by a factor of $\sim 50$ in scale\textsubscript{10} and a factor of $\sim 90$ in gridworld of the newly presented black-box strategies over the whitebox strategies, this is negligible.
\begin{table}[h]
	\begin{center}
		\begin{tabular}{l|n{5}{0}H|n{6}{0}H|n{5}{0}H|n{5}{0}H|HHn{2}{0}H|n{6}{0}H|n{6}{0}H|n{7}{0}H|}\toprule
			& \multicolumn{2}{l}{nand} & \multicolumn{2}{l}{bluetooth} & \multicolumn{2}{l}{scale\textsubscript{$10$}}&\multicolumn{2}{l}{crowds}&\multicolumn{2}{H}{Crowds}&\multicolumn{2}{l}{herman}&\multicolumn{2}{c}{\texttt{gw}}&\multicolumn{2}{c}{$\overline{\mathtt{gw}}$}&\multicolumn{1}{c}{$\overline{\mathtt{GW}}$}\\ \midrule
			
			$\#$ states & \multicolumn{1}{c}{7$\cdot$10$^7$} & & 143291 & & 121 & & \multicolumn{1}{c}{1$\cdot$10$^7$}& & 524288 & & \multicolumn{1}{c}{5$\cdot$10$^5$}& & 309327 & & 309327 & & \multicolumn{1}{c}{5$\cdot$10$^8$} & \\
			$c=1$& 31246 & 90 & 4428 & 14 & 116 &10&{\npboldmath} 44&1& {\npboldmath}437& 5&{\npboldmath}2& 1& 486&8 & 171219 &319 & 8082659&1207\\
			$c=2$& 18827 & 18 & 4548 & 9 & {\npboldmath} 75&4&61&1& 937& 5&{\npboldmath}1& 1&404&4 & {\npboldmath}152127 & 51 & 4883449 & 122\\
			$c=3$& 32777 & 10 & 7615 & 6 & 179&3&99&1& 6352& 5&{\npboldmath} 1& 1&{\npboldmath}293&3&579896&237&{\npboldmath}4252263 & 40\\ \midrule
			Bold\textsubscript{$0.1$}& 10583 & 6 & 4637 & 4 & 14528 & 1 & 199 & 1 & & &{\npboldmath} 0 & 0& \multicolumn{1}{c}{\texttt{TO}}& \multicolumn{1}{H|}{\texttt{TO}}& \multicolumn{1}{c}{\texttt{TO}}& \multicolumn{1}{H|}{\texttt{TO}}&\multicolumn{1}{H}{-} & \\
			Cautious\textsubscript{$10$} & {\npboldmath} 6900 & 5 & {\npboldmath} 2425 & 5 & 3670 &1 & 101 & 1 & & & \multicolumn{1}{c}{\texttt{TO}} & \multicolumn{1}{H|}{\texttt{TO}} & 26361 &  \multicolumn{1}{H|}{0.9}& \multicolumn{1}{c}{\texttt{TO}}& \multicolumn{1}{H|}{\texttt{TO}}&\multicolumn{1}{H}{-} & 
			\\ \bottomrule
		\end{tabular}
	\end{center}
	\caption{Average number of steps before the final restart, averaged over 300 (100 for Herman and $\overline{\mathtt{GW}}$) runs. Results for our strategies for $c=1,2,3$, and the bold and cautious strategies of  \cite{EKKW21}.}
	\label{tableconcur}
\end{table}
\begin{table}[h]
	\begin{center}
	\begin{tabular}{l|n{6}{0}l|n{4}{0}l|n{8}{0}l|n{8}{0}l|} \toprule
		& \multicolumn{2}{c}{AvC\textsubscript{17,8}(faulty)} &    \multicolumn{2}{c}{Maj\textsubscript{$\le 12$}(faulty)} & \multicolumn{2}{c}{AvC\textsubscript{17,8}}  & \multicolumn{2}{|c}{Maj\textsubscript{5,6}} \\ \midrule
		$c=1$ & 13645 & \texttt{ce} & 872 & \texttt{ce} &  126294 &\texttt{true}   & 4264508 & \texttt{ge}\\
		$c=2$ & 181746 & \texttt{ce} & 4763 &\texttt{ce} & 10485163 & \texttt{true}   & 11878533 & \texttt{ge} \\
		\midrule
		Peregrine & & \texttt{TO} & & \texttt{ce} & & \texttt{TO} &  & \texttt{true}
		\\ \bottomrule
	\end{tabular}	
	\end{center}
	\caption{Testing population protocols with the strategies $\mathfrak{S}[n\mapsto n^c]$. Experiments were run 100 times, averaging the number of steps to the last restart with a restart threshold of 250 for Average and Conquer (AvC) and \numprint{10000} for the Majority Protocol.}
	\label{tablepp}
	\vspace{-0.5cm}
\end{table}

\medskip\noindent\textbf{Experiments on population Protocols.} Population protocols are consensus protocols in which a crowd of indistinguishable agents decide a property of their initial configuration by reaching a stable consensus \cite{AADFP06,BlondinEJ18}. The specification states that for each initial configuration the agents eventually reach the right consensus (property holds/does not hold). We have tested our strategies on several protocols from the benchmark suite of Peregrine, the state-of-the-art model checker for population protocols \cite{BlondinEJ18,EsparzaHJM20}. The first protocol 
 of Table \ref{tablepp} is faulty, but Peregrine cannot prove it; our strategy finds initial configurations for which the protocol exhibits a fault. For the second protocol both our strategies and Peregrine find faulty configurations. The third protocol is correct; Peregrine fails to prove it, and our strategies correctly fail to find counterexamples.  The last protocol is correct, but in expectation consensus is reached only after an exponential number of steps in the parameters; we complement the specification, and search for a run that achieves consensus. Thanks to the logarithmic memory requirements, our strategies can run deep into the Markov chain and find the run.

\section{Conclusions}
We have studied the problem of testing partially observable stochastic systems against $\omega$-regular specifications in a black-box setting where testers can only restart the system, have no information on size or probabilities, and cannot observe the states of the system, only its outputs. We have shown that, despite these limitations, black-box testing strategies exist. We have obtained asymptotically optimal bounds on the number of steps to the last restart.
Surprisingly, our strategies never require many more steps than the strategies for fully observable systems of \cite{EKKW21}, and often even less. Sometimes, the improvement is by a large factor (up to $\sim 90$ in our experiments)
or the black-box strategies are able to solve instances where the strategies of \cite{EKKW21} time out.

\bibliographystyle{plain}
\bibliography{ref.bib}

\begin{thebibliography}{10}

\bibitem{AADFP06}
Dana Angluin, James Aspnes, Zo{\"{e}} Diamadi, Michael~J. Fischer, and
  Ren{\'{e}} Peralta.
\newblock Computation in networks of passively mobile finite-state sensors.
\newblock {\em Distributed Comput.}, 18(4):235--253, 2006.

\bibitem{Baier2008}
Christel Baier and Joost-Pieter Katoen.
\newblock {\em Principles of model checking}.
\newblock MIT Press, Cambridge, Massachusetts, 2008.

\bibitem{BasinJKZ13}
David~A. Basin, Vincent Jug{\'{e}}, Felix Klaedtke, and Eugen Zalinescu.
\newblock Enforceable security policies revisited.
\newblock {\em {ACM} Trans. Inf. Syst. Secur.}, 16(1):3:1--3:26, 2013.

\bibitem{BlondinEJ18}
Michael Blondin, Javier Esparza, and Stefan Jaax.
\newblock Peregrine: {A} tool for the analysis of population protocols.
\newblock In {\em {CAV} {(1)}}, volume 10981 of {\em Lecture Notes in Computer
  Science}, pages 604--611. Springer, 2018.

\bibitem{EsparzaHJM20}
Javier Esparza, Martin Helfrich, Stefan Jaax, and Philipp~J. Meyer.
\newblock Peregrine 2.0: Explaining correctness of population protocols through
  stage graphs.
\newblock In {\em {ATVA}}, volume 12302 of {\em Lecture Notes in Computer
  Science}, pages 550--556. Springer, 2020.

\bibitem{EKKW21}
Javier Esparza, Stefan Kiefer, Jan Kret{\'{\i}}nsk{\'{y}}, and Maximilian
  Weininger.
\newblock Enforcing {\(\omega\)}-regular properties in markov chains by
  restarting.
\newblock In {\em {CONCUR}}, volume 203 of {\em LIPIcs}, pages 5:1--5:22.
  Schloss Dagstuhl - Leibniz-Zentrum f{\"{u}}r Informatik, 2021.

\bibitem{FalconeMFR11}
Yli{\`{e}}s Falcone, Laurent Mounier, Jean{-}Claude Fernandez, and Jean{-}Luc
  Richier.
\newblock Runtime enforcement monitors: composition, synthesis, and enforcement
  abilities.
\newblock {\em Formal Methods Syst. Des.}, 38(3):223--262, 2011.

\bibitem{FalconeP19}
Yli{\`{e}}s Falcone and Srinivas Pinisetty.
\newblock On the runtime enforcement of timed properties.
\newblock In {\em {RV}}, volume 11757 of {\em Lecture Notes in Computer
  Science}, pages 48--69. Springer, 2019.

\bibitem{DBLP:journals/corr/abs-2002-07080}
Christian Hensel, Sebastian Junges, Joost{-}Pieter Katoen, Tim Quatmann, and
  Matthias Volk.
\newblock The probabilistic model checker storm.
\newblock {\em CoRR}, abs/2002.07080, 2020.

\bibitem{DBLP:conf/qest/KwiatkowskaNP12}
Marta~Z. Kwiatkowska, Gethin Norman, and David Parker.
\newblock The {PRISM} benchmark suite.
\newblock In {\em {QEST}}, pages 203--204. {IEEE} Computer Society, 2012.

\bibitem{LarsenL16a}
Kim~G. Larsen and Axel Legay.
\newblock On the power of statistical model checking.
\newblock In {\em ISoLA {(2)}}, volume 9953 of {\em Lecture Notes in Computer
  Science}, pages 843--862, 2016.

\bibitem{LeeY96}
David Lee and Mihalis Yannakakis.
\newblock Principles and methods of testing finite state machines-a survey.
\newblock {\em Proc. {IEEE}}, 84(8):1090--1123, 1996.

\bibitem{LegayDB10}
Axel Legay, Beno{\^{\i}}t Delahaye, and Saddek Bensalem.
\newblock Statistical model checking: An overview.
\newblock In {\em {RV}}, volume 6418 of {\em Lecture Notes in Computer
  Science}, pages 122--135. Springer, 2010.

\bibitem{LigattiBW09}
Jay Ligatti, Lujo Bauer, and David Walker.
\newblock Run-time enforcement of nonsafety policies.
\newblock {\em {ACM} Trans. Inf. Syst. Secur.}, 12(3):19:1--19:41, 2009.

\bibitem{PeledVY02}
Doron~A. Peled, Moshe~Y. Vardi, and Mihalis Yannakakis.
\newblock Black box checking.
\newblock {\em J. Autom. Lang. Comb.}, 7(2):225--246, 2002.

\bibitem{Schneider00}
Fred~B. Schneider.
\newblock Enforceable security policies.
\newblock {\em {ACM} Trans. Inf. Syst. Secur.}, 3(1):30--50, 2000.

\bibitem{Sen04}
Koushik Sen, Mahesh Viswanathan, and Gul Agha.
\newblock Statistical model checking of black-box probabilistic systems.
\newblock In {\em CAV}, pages 202--215, 2004.

\bibitem{SenVA05}
Koushik Sen, Mahesh Viswanathan, and Gul Agha.
\newblock On statistical model checking of stochastic systems.
\newblock In {\em {CAV}}, volume 3576 of {\em Lecture Notes in Computer
  Science}, pages 266--280. Springer, 2005.

\bibitem{We22}
Maximilian Weininger.
\newblock Personal communication, 2022.

\bibitem{Younes05}
H{\aa}kan L.~S. Younes.
\newblock Probabilistic verification for "black-box" systems.
\newblock In {\em {CAV}}, volume 3576 of {\em Lecture Notes in Computer
  Science}, pages 253--265. Springer, 2005.

\bibitem{YounesCZ10}
H{\aa}kan L.~S. Younes, Edmund~M. Clarke, and Paolo Zuliani.
\newblock Statistical verification of probabilistic properties with unbounded
  until.
\newblock In {\em {SBMF}}, volume 6527 of {\em Lecture Notes in Computer
  Science}, pages 144--160. Springer, 2010.

\bibitem{YounesS02}
H{\aa}kan L.~S. Younes and Reid~G. Simmons.
\newblock Probabilistic verification of discrete event systems using acceptance
  sampling.
\newblock In {\em {CAV}}, volume 2404 of {\em Lecture Notes in Computer
  Science}, pages 223--235. Springer, 2002.

\end{thebibliography}
\appendix

\section{Proofs of Section \ref{sec:blackboxtesting}}

\rabinlemma*

\begin{proof}
Let $\dra = (\draS, \Sigma, \draTr, \draInit, \draAcc)$ be a \DRA{} recognizing $\lang \subseteq \Sigma^\omega$, and assume that $\draAcc=\{(E_1, F_1), \ldots, (E_k, F_k)\}$. Let $\strategy_k$ be a black-box strategy for the Rabin language $\rabin_k$.  We construct a black-box strategy $\strategy_\lang$ for $\lang$. 

Let $w = \ell_1 a_1 \ell_2 \cdots \ell_{n-1} a_n \ell_n \in (\Sigma \times \{\restact, \mathsf{c}\})^* \Sigma$. We define the action $\strategy_\lang(w)$ as follows. Let $\draInit q_1 \ldots q_n$ be the unique run of $\dra$ on the word $\ell_1\ell_2 \ldots \ell_n \in \Sigma^*$. Define $v = \ell_1' a_1 \ell_2' \cdots \ell_{n-1}' a_n \ell_n' \in (2^{M_k} \times \{\restact, \mathsf{c}\})^* 2^{M_k}$ as the word given by: $\infl_j \in \ell_i'$ if{}f $q_i \in E_j$, and  $\finl_j \in \ell_i'$ if{}f $q_i \in F_j$. (Intuitively, we mark with $\infl_j$ the positions in the run at which the \DRA{} visits $E_j$, and with $\finl_j$ the positions at which the \DRA{} visits $F_j$.) We set $\strategy_\lang(w) \coloneqq  \strategy_k(v)$.  

We claim that $\strategy_\lang$ is a black-box strategy for $\lang$. 
To prove this claim, let $\Mc = (\St, \init, \Sigma, \Lab, \Pm)$ be an arbitrary Markov chain with labels in $\Sigma$.
Define the product of $\Mc$ and $\mathcal{A}$ as the labeled Markov chain $\Mcdra =(\St \times \draS, \init', 2^{M_k}, \Lab', \Pm')$, where 
\begin{itemize}
\item $\init' = (\init, \draInit)$;  
\item $\infl_j \in \Lab'(s,q)$ if if{}f $q \in E_j$ and  $\finl_j \in \Lab'(s,q)$ if{}f $q \in F_j$;
\item $\Pm'((s,q),(s',q')) = \Pm(s,s')$ if $q'=\draTr(q,\Lab(s))$ and $0$ otherwise.
\end{itemize} 

Since $\mathcal{A}$ is deterministic, for every run $\pi =s_{in} s_1 s_2 \cdots$ of $\Mc$ there exists a unique run $\pi'=(s_{in} , \draInit) (s_1, q_1) (s_2, q_2) \cdots$ of $\Mcdra$, and the mapping that assigns $\pi'$ to $\pi$ is a bijection. By the definition of the accepting runs of $\mathcal{A}$, we have $\Lab(\pi) \in \lang$ if{}f $\draInit q_1 q_2 \cdots$ is an accepting run of $\mathcal{A}$ and, by the definition of the product, if{}f $\Lab'(\pi') \in \rabin_k$. Further, by the definition of $\Pm'$, we have $\Pr_{\Mc}(L) = \Pr_{\Mcdra}(L_C)$. 

Consider now the Markov chains $\restartMDPst{\strategy_k}$ and $(\MDP \otimes \dra)_r^{\strategy_\lang}$.
A run of $\restartMDPst{\strategy_k}$ can be seen as an infinite sequence $\Pi = \pi_0 \, r \, \pi_1 \, r \cdots$, where $\pi_0, \pi_1, \cdots$ are paths of $\Mc$ indicating that $\Mc$ is restarted after executing $\pi_0, \pi_q$ etc., and similarly for $(\MDP \otimes \dra)_r^{\strategy_\lang}$. (We omit the occurrences of the continue action $c$.) We extend the mapping above so that it assigns to $\Pi$ the run $\Pi' = \pi_0' \, r \, \pi_0' \, r \, \cdots$. We have that $\Pi$ is a run of $\restartMDPst{\strategy_k}$ satisfying  $\tail{\Pi} \in \lang$ if{}f $\Pi'$ is a run of $(\MDP \otimes \dra)_r^{\strategy_\lang}$ satisfying  $\tail{\Pi'} \in \rabin_k$. Further, since the probabilities of the transitions in the run coincide,  $\strategy_\lang$ is a black-box strategy for $\lang$, and the claim is proved.
\end{proof}

\section{A general definition for the progress radius and probability}

\label{appendix:progress}
We will now restate the definitions for the progress radius and probability for Rabin languages with more than one Rabin pair. They only differ by some technicalities from the definitions given in Section \ref{subsec:progress}. For convenience, we have underline all the differences.

\medskip\noindent\textbf{Good runs and good \BSCC{}s.} We extend the definition of good paths to good runs and good \BSCC{}s of a Markov chain.
A run $\run=s_0s_1s_2\dots$ is \emph{good} if \ul{there exists a Rabin pair $(\infl_i, \finl_i)$ such that $\infl_i$ appears infinitely often in $\run$ and $\finl_i$ finitely often}, and \emph{bad} otherwise. 
So a run $\run$ is good if{}f there exists a decomposition of $\run$  into an infinite concatenation $\run \coloneqq  \mypath_0 \odot \mypath_1 \odot \mypath_2 \odot \cdots$ of non-empty paths \ul{such that there exists an $1\le i\le k$ such that $\mypath_1, \mypath_2, \ldots$ are $i$-good}. We let $\Pg$ denote the probability of the good runs of $\Mc$.

A \BSCC{} of $\Mc$ is \ul{$i$}-\emph{good} if it contains at least one state labeled by $\infl_i$ and no state labeled by $\finl_i$. If a \BSCC{} is not \ul{$i$-good for any $1\le i\le k$} we call it bad.

\begin{definition}[Good-reachability and good-witness radii]
	Let $B_\gamma$ be the set of states of $\Mc$ that belong to good \BSCC{}s and let $\St_\gamma$ be the set of states from which it is possible to reach $B_\gamma$
	and let $s \in \St_\gamma$. A non-empty path $\pi$ starting at $s$ is a \emph{good progress path} if 
	\begin{itemize}
		\item $s \in \St_\gamma \setminus B_\gamma$, and $\pi$ ends at a state of $B_\gamma$; or
		\item $s \in B_\gamma$, and $\pi$ ends at a state with observation \ul{$\infl_i$ and $s$ is in an $i$-good BSCC}.  
	\end{itemize}
	The \emph{good-reachability radius} $r_\gamma$ is the maximum, taken over every $s \in \St_\gamma \setminus B_\gamma$, of the length of a shortest progress path for $s$.
	The \emph{good-witness radius} $R_\gamma$ is the same maximum, but taken over every $s \in B_\gamma$.
\end{definition}

The bad-reachability and bad-witness radii, denoted $r_\beta$ and $R_\beta$ are defined similarly. Only the notion of progress path of a state$s \in B_\beta$ 
needs to be adapted. Loosely speaking, \ul{for every state with observation $\finl_i$ a bad BSCC contains at least one state with observation $\infl_i$.}
Accordingly, if no state of the \BSCC{} of $s$ has an observation $\infl_i$ \ul{for any $i$}, then any non-empty path starting at $s$ is a progress path.
and otherwise a progress path of $s$ is a non-empty path starting at $s$ that, \ul{for every state with observation $\infl_i$ in the BSCC of $s$, contains a state with observation $\finl_i$.}
In other words, a progress path starting in a bad \BSCC{} $B$ visits states with \ul{all observations $\finl_i$ that prevent $B$ from being a good BSCC.} Note that we leave the bad progress radii and probabilities undefined if the chain does not contain a bad \BSCC{}, and hence runs are good w.p.1.

\begin{definition}[Progress radius]
	The \emph{progress radius} $\Rm$ of $\Mc$ is the maximum of $r_\gamma$, $R_\gamma$, $r_\beta$, and $R_\beta$. 
\end{definition}

\medskip\noindent\textbf{Progress probability.} The progress probability is now defined in the same way as it is done in the main part of the paper. From any state of the Markov chain it is possible to ``make progress'' by executing a progress path of length $\Rm$. 
However, the probability of such paths varies from state to state. Intuitively, the progress probability gives a lower bound on the probability
of making progress. 

\begin{definition}
	Let $B_\gamma$ be the set of states of $\Mc$ that belong to good \BSCC{}s, let $\St_\gamma$ be the set of states from which it is possible to reach $B_\gamma$
	and let $s \in \St_\gamma$. The \emph{good-reachability probability} $p_\gamma$ is the minimum, taken over every $s \in \St_\gamma\setminus B_\gamma$, of the probability 
	that a path with length $r_\gamma$ starting at $s$ contains a good progress path. The \emph{good-witness probability} $P_\gamma$ is the same minimum, but taken over every $s \in B_\gamma$ with paths of length $R_\gamma$.
	The corresponding bad probabilities are defined analogously. The \emph{progress probability} $\Pmax$ is the minimum of $p_\gamma, P_\gamma, p_\beta, P_\beta$.
\end{definition}

\section{Proofs of Section \ref{sec:strategy}}
In this section, we will give the technical proofs omitted in the main paper.
\ProbRestartGood*

\begin{proof}
	If $\lfloor f(n)/\Rm \rfloor \leq 1$ then the inequality holds trivially. So assume $f(n) \geq 2 \Rm $.
	
	
	Let $\run \in \NB_n$. Observe that $\run$ eventually reaches a \BSCC{} of $\Mc$ w.p.1 and, since $\run$ only visits states of $S_\gamma$, that \BSCC{} is good. Let $B$ be this \BSCC{}. Assume that $\restart(\run) \geq n$. We consider the following cases, where we start counting steps immediately after the $(n-1)$-th restart and, for $a <b$, the path $[a, b]$ is  the path that starts immediately before step $a$, and ends immediately after step $b$.
	\begin{itemize}
		\item After $2f(n)$ steps, $\run$ has not yet reached $\mathcal{B}$. \\
		By the definition of $p_\gamma$, this happens with probability at most $(1-p_\gamma)^{f(n)/\RgBSSC}$.
		\item After $2f(n)$ steps, $\run$ has already reached $\mathcal{B}$. Further, the $n$-th restart happens 
		in the path $[2 f(n), 2(f(n)+R_\gamma)-1]$. \\
		In this case, by the definition of $\mathfrak{S}[f]$, the second half of the last path sample does  
		does not contain any state labelled with $e_i$ such that the \BSCC{} is $i$-good. It follows that the path $[f(n)+R_\gamma+1, 2f(n)]$ does not visit $W_B$. By the definition of $P_\gamma$ and $R_\gamma$, this happens with probability at most $(1 -P_\gamma)^{\lfloor f(n)/R_\gamma\rfloor-1}$. 
		\item After $f(n)$ steps, $\run$ has already reached $\mathcal{B}$. Further, the $n$-th restart happens 
		after the step $2f(n)+2R_\gamma-1$. \\
		In this case we let $k \ge\lfloor f(n)/R_\gamma\rfloor+1$ be the smallest number such that the path $[(k+1)R_\gamma+1, 2kR_\gamma]$ does not contain any witness states, i.e.~ states labelled with $\infl_i$.

	\end{itemize}
	
	By the definition of $\mathfrak{S}[f]$, if $\run$ restarts in the interval  $[2lR_\gamma+1, 2(l+2)R_\gamma-1]$ 
	 and has reached a good \BSCC{} $\mathcal{B}$ in the first $f(n)$ steps, then it is covered by the third case for some $k$ with $\lfloor f(n)/R_\gamma\rfloor+1\le k \le l$. Because $k$ is the smallest $k$ satisfying this property, and we are not in the second case, the run performed a progress path of $\mathcal{B}$ between step $(l-1)R_\gamma+1$ and $lR_\gamma$, otherwise we would have already counted this case. Hence we can bound the sum of probabilities of the last two cases by $(1-P_\gamma)^{\lfloor f(n)/R_\gamma\rfloor-1}
		+ \sum_{k=\lfloor f(n)/R_\gamma\rfloor-1}^\infty P_\gamma(1-P_\gamma)^k$.
	

So we get:

\begin{align*}
	& Pr[\restart \ge n \mid \NB_n] \\
	\le & (1-p_\gamma)^{\lfloor f(n)/r_\gamma\rfloor}+(1-P_\gamma)^{\lfloor f(n)/R_\gamma\rfloor-1}
	+ \sum_{k=\lfloor f(n)/R_\gamma\rfloor-1}^\infty P_\gamma(1-P_\gamma)^k\\
	\le & 3(1-\Pmax )^{\lfloor f(n)/\Rm \rfloor-1}
\end{align*}
\end{proof}

\ProbRestart*

\begin{proof}
	We have 
	\begin{align*}
		\Pr [\restart\ge n\mid R\ge n-1] = & \Pr[\restart \ge n \mid \NB_{n}]  \cdot \Pr[\NB_{n} \mid R \ge n-1] + \\
		& \Pr[\restart \ge n \mid \overline{\NB}_n] \cdot \Pr[\overline{\NB}_n \mid \restart \ge n-1].
	\end{align*}
	Let $\alpha \coloneqq  3(1-\Pmax )^{\lfloor f(n)/\Rm \rfloor-1}$.
	Applying  Lemma \ref{lemma:PRestartGood} and $\Pr[\restart \ge n \mid \overline{\NB}_n]\le 1$, we get
\begin{align*}
                    & \Pr [\restart\ge n\mid \restart\ge n-1]  \\
		 \leq \; & \alpha  \Pr[\NB_{n} \mid \restart \ge n-1] + 
		 \Pr[\overline{\NB}_n \mid \restart \ge n-1] \\
		 \leq \; & \alpha  \Pr[\NB_{n} \mid \restart \ge n-1] + 
		 (1 - \Pr[\NB_n \mid \restart \ge n-1]) \\
		 \leq \; & (\alpha -1)   \Pr[\NB_{n} \mid \restart \ge n-1] + 
		 1	  
\end{align*}

W.p.1, good runs of $\Mc$ only visit states of $S_\gamma$. (Indeed, if a good run visits some state outside $S_\gamma$, then the run can only reach a bad BSSC. Since the run is good, the run cannot visit any BSSC at all, which can only happen with probability $0$.) Hence, the probability to only visit states of $S_\gamma$ before restarting is at least $\Pg$ for arbitrary strategies, i.e. $\Pr[\NB_{n} \mid \restart \ge n-1] \ge \Pg$. It follows 
	\begin{align*}
		\Pr [\restart\ge n\mid \restart\ge n-1] \le \; &  (\alpha -1) \Pg +1 = 1 - \Pg ( 1 - \alpha) 
	\end{align*}
\end{proof}
\section{Proofs of Section \ref{sec:quantitative}}

We prove Lemma \ref{lemma:boundnsmall}. We need a technical result:

\begin{lemma}[A technical lemma]
	\label{lemma:sumnsquaredinfinity}
	\label{lemma:sumnsquared}
	\label{lemma:verygoodlemma}
	For $c,X\in \Z_{\ge 0}$ and $0<p<1$ we have that
	\[
	\sumi_{n=X} n^c\cdot p^{n-X}\le (c+1)!\left(\frac{(X+c)^c}{1-p}+\frac{1}{(1-p)^{c+1}}\right).
	\]
\end{lemma}
\begin{proof}
	We will proof the lemma by induction on $c$ starting with $c=0$. Then we have by the formula for the geometric progression that:
	\[
	\sumi_{n=X} p^{n-X}=\sumi_{i=0} p^i = \frac{1}{1-p}.
	\]
	This proves the the induction base case.
	Now assume we have proven 
	\[
	S_{c-1}(X)\coloneqq \sumi_{n=X} n^{c-1}\cdot p^{n-X}\le c!\left(\frac{(X+c-1)^{c-1}}{1-p}+\frac{1}{(1-p)^{c}}\right).
	\]
	Now consider 
	\[
	S_c(X)\coloneqq \sumi_{n=X} n^c\cdot p^{n-X}
	\]
	When multiplying by $(1-p)$ we get the following
	\begin{align*}
		(1-p)S_c(X)&=X^c+\sumi_{n=X} \left((n+1)^c-n^c\right)\cdot p^{n-X}\\
		&=X^c+\sumi_{n=X} \left(\sum_{k=1}^{c} \binom{c}{k}n^{c-k}\right)p^{n-X}\\
		&\le X^c+\sumi_{n=X}c\left(\sum_{k=0}^{c-1}\binom{c-1}{k}n^{c-k}\right)p^{n-X}\\
		&= X^c+\sumi_{n=X}c(n+1)^{c-1}p^{n-X}\\
		&= X^c+ c\cdot \sumi_{n=X+1}n^{c-1}p^{n-X-1}\\
		&=X^c+ cS_{c-1}(X+1)\\
		&\le X^c +c \cdot c!\left(\frac{(X+c)^{c-1}}{1-p}+\frac{1}{(1-p)^{c}}\right)\\
		S_c(X)&\le (c+1)!\left(\frac{(X+c)^c}{1-p}+\frac{1}{(1-p)^{c+1}}\right)
	\end{align*}
	This concludes the proof.
\end{proof}

\stepsfragment*
\begin{proof}
	\label{proof:stepsfragment}
	We consider three cases:
\begin{itemize}
\item[(1)] The run $\rho$ gets restarted for the $n$-th time in at most $2f(n)+2\Rm $ steps after the $(n-1)$-st restart. We can bound the expected number of steps in this case by $2f(n)+2\Rm $.
\item[(2)] The run $\rho$ executes at least step $2f(n)+2\Rm +1$ after the $(n-1)$-st restart without another restart, and only visits states in $S_\beta$. Then, the expected number of steps until a bad \BSCC{} $\mathcal{B}$ is reached is equal to $r_\gamma/p_\gamma$. After another $r_\gamma/p_\gamma$ steps, the entire second half of states visited since the last restart is now contained in $\mathcal{B}$. Then, the expected number of steps required to perform a progress path of the \BSCC{} is $R_\beta/P_\beta$. After at most $2f(n)$ additional steps, the strategy restarts. Hence, an upper bound of the expected number of steps in this case is $2r_\beta/p_\beta+R_\beta/P_\beta+2f(n)$. 
		
\item[(3)] The run $\rho$ reaches at least step $2f(n)+2\Rm +1$ after the $(n-1)$-st restart without another restart, and only visits states in $S_\gamma$. In this case it takes on average at most $2r_\gamma/p_\gamma$ steps to reach a good \BSCC{}. After that, we divide the rest of the run into blocks of length $2R_\gamma$. Let $v$ be the number such that the restart happens between steps $2vR_\gamma$ and $2(v+1)R_\gamma$. Then we have:
\begin{itemize}
\item[(a)]$\run$ visits an accepting state of the \BSCC{} between steps $(v-1)R_\gamma$ and $(v+1)R_\gamma$. \\
Otherwise the restart happens before step $2vR_\gamma$
\item[(b)] $\run$ does not visit accepting states of the \BSCC{} between step $(v+1)R_\gamma$ and $2vR_\gamma$. \\
Indeed, if we restart at step $2vR_\gamma$, then we have not visited an accepting state of the good \BSCC{} between step $vR_\gamma$ and $2vR_\gamma$. If we restart at step $2(v+1)R_\gamma$, the same applies to step $(v+1)R_\gamma$ and $2(v+1)R_\gamma$. For restarts between steps $2vR_\gamma$ and $2(v+1)R_\gamma$ a 
corresponding in-between statement is true. In all these cases the run never visits an accepting state of the \BSCC{}
between step $(v+1)R_\gamma$ and $2vR_\gamma$.  
\end{itemize}
The probability of (a) is at least $2P_\gamma$, and the probability of (b) is at most $(1-P_\gamma)^{v-1}$. But if (a) was not the case, we would have already counted it with probability $(1-P_\gamma)$ with at most $2R_\gamma$ steps less.
So the expected number of steps in the cases, in which we execute at most $2(v+1)R_\gamma$ steps after the $(n-1)$-st restart for some $v$, is bounded by the sum over $2P_\gamma (1-P_\gamma)^{v-1}$ mulitplied by the number of steps for all possible values of $v$.
\begin{equation*}
2f(n)+2\Rm +\sum_{v=1}^\infty 4 R_\gamma (v+1)P_\gamma(1-P_\gamma)^{v-1}=2f(n)+2\Rm +\frac{4R_\gamma}{P_\gamma(1-P_\gamma)}.
\end{equation*}
\end{itemize}
\noindent We obtain 
\begin{align*}
E[S_n\mid \restart\ge n-1] = \; & \big(2f(n)+2\Rm  \big) p_1 + \\
                                       & 2r_\beta/p_\beta+R_\beta/P_\beta+2f(n) p_2 + \\
                                       & 2r_\gamma/p_\gamma + \frac{4R_\gamma}{P_\gamma(1-P_\gamma)} +\left(2f(n)+2R_\gamma\right) (1- p_1 - p_2)
\end{align*}
where $p_1, p_2$ are the probabilities of (1) and (2), respectively. Using $p_1, p_2 \leq 1$ and simple arithemtic yields the generous bound:
	\begin{equation*}
		E[S_n\mid \restart\ge n-1]\le 2 (\Rm+f(n))+9\left(\frac{\Rm }{\Pmax (1-P_\gamma)}\right)
	\end{equation*}
\end{proof}

%

\totalsteps*
\begin{proof}
	\label{proof:totalsteps}
	By linearity of expectation, we have $\expected[S]=\sum_{i=1}^\infty \expected(S_n)$. The idea of the proof is to split the sum into two parts: for $n<X$, and for $n\ge X$. For $n<X$ we just approximate $\Pr[\restart\ge n-1]$ by $1$. For $n>X$ we can say more thanks to Lemma \ref{lemma:ChoiceOfX}: 
	\begin{align*}
		\Pr[\restart\ge n-1]&=\Pr[\restart\ge n-1\mid \restart\ge n-2] \dots \Pr[\restart\ge X+1\mid \restart\ge X]\cdot \Pr[\restart\ge X]\\
		&\le \prod_{k=\lceil X\rceil}^n \Pr[\restart\ge k\mid \restart\ge k-1]\\
		&\le \left(1-\Pg/2\right)^{n-X}
	\end{align*}
	This yields:
	\begin{align*}
		E[S]&=\sum_{n=0}^\infty E[S_n\mid \restart\ge n-1]\Pr[\restart\ge n-1]\\
		&\le \sum_{n=0}^XE[S_n\mid \restart\ge n-1]+
		\sum_{n=X}^\infty E[S_n\mid \restart\ge n-1]\cdot \left(1-\Pg/2\right)^{n-X}
	\end{align*}
	We bound the first summand applying Lemma \ref{lemma:boundnsmall}:
	\begin{align*}
	& \sum_{n=0}^X E[S_n\mid \restart\ge n-1] \\
		\le & \sum^X_{n=0}2n^c+2\Rm+9\left(\frac{\Rm }{\Pmax (1-P_\gamma)}\right)\le 2X^{c+1}+X\Rm \left(2+9\frac{1}{\Pmax (1-P_\gamma)}\right).
	\end{align*}
	Now bound the second summand, applying Lemma \ref{lemma:boundnsmall} again:
	\begin{align*}
		& \phantom{\le \; } \sum_{n=X}^\infty E[S_n\mid \restart\ge n-1]\cdot \left(1-\Pg/2\right)^{n-X} \\	
		& \le \sumi_{n=X}\left(2n^c+2\Rm+9\left(\frac{\Rm }{\Pmax (1-P_\gamma)}\right)\right)\left(1-\Pg/2\right)^{n-X}\\
		&\le \frac{2 \Rm }{\Pg} \left( 2+9\frac{1}{\Pmax (1-P_\gamma)}\right)+ 2\sumi_{n=X} n^c \left(1-\Pg/2\right)^{n-X}\\
		&\le\frac{2 \Rm }{\Pg} \left( 2+9\frac{1}{\Pmax (1-P_\gamma)}\right)+ 2(c+1)!\left(\frac{2(X+c)^c}{\Pg}+\frac{2^{c+1}}{\Pg^{c+1}}\right)
	\end{align*}
	where in the last step we used Lemma \ref{lemma:verygoodlemma}. 
	\begin{equation*}
		\Esteps \in \bigO \left((c+1)!\cdot 2^c\cdot \left(\frac{\Rm }{\Pmax }\right)^{1+1/c}+\frac{2^c(c+1)!}{P_{good}^{c+1}} + (c+1)!(2c)^{c+1}\right)
	\end{equation*}
	For a fixed value of $c$, i.e., for the specific strategy $f(n)=n^c$, this bound simplifies to
	\begin{equation*}
		\Esteps \in \bigO \left(\left(\frac{\Rm }{\Pmax }\right)^{1+1/c}+\frac{1}{\Pg^{c+1}}\right).
	\end{equation*}
\end{proof}
\end{document}